\documentclass[11pt]{article}
\usepackage{amsmath, amssymb, amsthm, mathtools}
\usepackage{bm}
\usepackage{natbib}
\usepackage[
  left=2.3cm,
  right=2.3cm,
  top=2.5cm,
  bottom=2.5cm
]{geometry}
\usepackage{appendix}
\newtheorem{theorem}{Theorem}[section]
\newtheorem{lemma}[theorem]{Lemma}
\newtheorem{corollary}[theorem]{Corollary}
\newtheorem{proposition}[theorem]{Proposition}
\newtheorem{definition}[theorem]{Definition}
\newtheorem{remark}[theorem]{Remark}

\newcommand{\ind}{\mathds{1}}

\usepackage[utf8]{inputenc}
\usepackage{natbib}
\usepackage{booktabs}
\usepackage{longtable}
\usepackage{lineno,hyperref}
\usepackage{lscape}
\usepackage{xcolor}
\usepackage{tabularx}
\usepackage{threeparttable}
\usepackage{color}
\usepackage{graphics}
\usepackage{graphicx} 
\usepackage{diagbox}
\usepackage{colortbl}
\usepackage{setspace}
\usepackage{amssymb}
\usepackage{amsmath}
\usepackage{mathrsfs}
\usepackage{calrsfs}
\usepackage{multirow}
\usepackage{rotating}
\usepackage{pdflscape}
\usepackage{dsfont}
\usepackage{stmaryrd}
\usepackage{float}
\usepackage{lmodern} 
\usepackage{algorithm}
\usepackage{algpseudocode}
\usepackage{comment}
\usepackage[commentColor=cyan]{algpseudocodex}

\DeclareMathOperator*{\argmin}{arg\,min}
\DeclareMathAlphabet{\pazccal}{OMS}{zplm}{m}{n}

\usepackage{enumitem}
\usepackage{lmodern}

\begin{document}

\title{(In)Efficient Market States and Rough Volatility Detected via
Gr\"{u}nwald-Letnikov Fractional Derivative}
\author{Daniele Angelini\\ \small{MEMOTEF, Sapienza University of Rome (Italy)}\\\small{daniele.angelini@uniroma1.it} }
\date{}

\maketitle
\begin{abstract}
Testing self-similarity in fractional processes from a single observed trajectory is difficult under long-range dependence, because the associated Kolmogorov--Smirnov (KS) statistic undergoes a phase transition when $H>1/2$. In this regime, the classical limit collapses to a non-functional absolute Gaussian law and finite-sample convergence becomes severely distorted. This paper introduces a regime-adaptive KS/GL--KS framework based on the discrete Gr\"{u}nwald--Letnikov (GL) fractional derivative. The GL filter removes the low-frequency long-memory singularity while preserving the finite-dimensional $H$-self-similarity needed for distributional identification. We derive the filtered empirical-process limit, prove consistency and local asymptotic behavior of the resulting Hurst estimator, and validate the method through Monte Carlo simulations. Financial applications to realized volatility and equity index prices show how the procedure detects rough volatility and persistent, anti-persistent, or efficient market states.
\end{abstract}

\textbf{Keywords:} Hurst exponent; Fractional processes; Kolmogorov--Smirnov test; Gr\"{u}nwald-Letnikov derivative

\section{Introduction}\label{sec:intro}

Fractional Brownian motion (fBm), pioneered by Kolmogorov \citep{Kolmogorov1940} and formalized by Mandelbrot and Van Ness \citep{MandelbrotVanNess1968}, stands as a cornerstone stochastic framework for modeling complex systems that exhibit self-similarity, non-stationarity, and scaling behaviors. Fully characterized by its Hurst parameter $H\in(0,1]$, fBm has become indispensable across a multitude of quantitative disciplines, most notably in financial economics, where asset prices, exchange rates, and volatility dynamics frequently display persistent scaling laws. The increment process of fBm, known as fractional Gaussian noise (fGn), inherits this scaling behavior and exhibits a distinctive memory structure governed entirely by $H$. When $H=1/2$, fGn reduces to standard independent white noise; for $H<1/2$, it features anti-persistence or short-range dependence (SRD). Conversely, when $H>1/2$, the process enters the long-range dependence (LRD) regime, where its autocovariance function decays non-integrably and its spectral density diverges at the origin.

In empirical applications, precisely identifying and validating the true self-similarity parameter $H$ is a critical inferential task. Traditional estimation approaches for the Hurst exponent and long-range dependence are often
based on moment or second-order scaling relations, including variance plots, rescaled-range methods, absolute-moment scaling, DFA-type procedures, and periodogram regression
\cite{TaqquTeverovskyWillinger1995,Weron2002}. Although these methods are widely used, their finite-sample behavior can be strongly affected by short records, trends, shifts in the mean, and scaling crossovers \cite{TeverovskyTaqqu1997,Kantelhardt2001}. Such effects may generate spurious evidence of persistence or distort the estimated scaling exponent, especially in financial series where structural breaks and heavy-tailed fluctuations are common. To mitigate these vulnerabilities, recent advancements have introduced non-parametric, distribution-based approaches to evaluate self-similarity. By exploiting the strict invariance of rescaled increment distributions across varying time scales, one can define the diameter of a family of distributions, which theoretically achieves a global minimum of zero at the true parameter $H$. When applied to empirical data, this framework constructs a two-sample Kolmogorov--Smirnov (KS) type statistic from lagged increments extracted from a single observed trajectory \citep{Bianchi2004}.

However, applying this framework to time series encounters significant theoretical obstacles. Specifically, extracting rescaled processes from a single sample trajectory introduces a complex dependence structure that violates the core independence assumptions underlying the classical KS test. This issue manifests on two distinct fronts: the \textit{intradependence problem}, which stems from the internal autocorrelation of the process governed by the self-similarity parameter $H$, and the \textit{interdependence problem}, which arises from the statistical overlap among the extracted rescaled series \citep{AngeliniBianchi2025}. Consequently, the standard KS distribution proves inadequate, necessitating an alternative asymptotic distribution that explicitly accounts for the $H$-induced dependence structure. While \citep{AngeliniBianchi2025} partially mitigated these hurdles using a random permutation approach, this method remains computationally intensive and, crucially, introduces an additional source of estimation variance, thereby highlighting the need for an accelerated, robust analytical framework.

However, this transition is not merely computational; it is fundamentally dictated by the underlying stochastic properties of the process. As established in Propositions \ref{prop:SRD} and \ref{prop:LRD}, the asymptotic behavior of this KS-type statistic undergoes a sharp, structurally disruptive phase transition at the critical threshold of $H=0.1/2$. In the SRD and independent regimes ($H \leq 1/2$), the statistic converges swiftly to a theoretical limit defined by the supremum of a centered Gaussian process with an explicit covariance structure. Conversely, when $H > 1/2$ the LRD regime necessitates a distinct normalization, as the limit collapses to a scaled absolute Gaussian random variable. As shown in Remark \ref{rmk:slowLRD}, the persistent, heavy correlation structures inherent to the LRD regime generate severe finite-sample distortions, leading to a slow rate of empirical convergence. This slow convergence fundamentally devalues the standard asymptotic distribution for $H > 1/2$, rendering the classical framework practically ineffective for robust hypothesis testing and goodness-of-fit verification in long-memory environments.

To overcome this issue, this paper presents a novel mathematical architecture designed to accelerate convergence by separating the long-range memory of the process from its self-similar scaling properties. We achieve this by filtering the fGn process through a discrete fractional transformation based on the Gr\"{u}nwald-Letnikov (GL) derivative. By applying a discrete GL derivative of order $\alpha$, where $\alpha$ is strategically chosen such that $H-\alpha<1/2$, we fundamentally re-engineer the low-frequency behavior of the stochastic process. In the frequency domain, while the spectral density of the original fGn blows up at the origin due to long memory, the spectral density of the GL-filtered process behaves asymptotically as $f(\lambda)\sim |\lambda|^{1-2(H-\alpha)}$ as $\lambda\to 0$. Because $H-\alpha<1/2$, the exponent becomes positive, effectively shifting the filtered process out of the LRD domain and into a short-memory, anti-persistent regime. Crucially, we prove that this fractional transformation decouples memory from scaling without disrupting the underlying self-similarity relation required for distributional self-similarity testing; the discrete GL derivative preserves the strict $H$-self-similarity in the sense of finite-dimensional distributions across distinct scales.

Building upon this theoretical mechanism, this paper pursues two main aims:
\begin{itemize}
    \item Derivation of the Asymptotic Distribution: Our first objective is to find and formalize the limiting distribution of the newly proposed Gr\"{u}nwald--Letnikov--Kolmogorov-Smirnov (GL-KS) test statistic for this transformed short-memory process. By successfully neutralizing the LRD via the GL filter, we bypass the correlation-induced slow convergence rates that plague the classical test. We demonstrate that the empirical process of the GL-derivative adheres to short-range dependence functional central limit theorems, enabling the GL-KS statistic to converge rapidly to the supremum of a well-defined, centered Gaussian process. This provides a robust, computationally viable diagnostic tool with stable critical values, even when the underlying original process exhibits extreme long memory.
    \item Application to Finance: Our second objective is to showcase the practical utility of this accelerated framework within financial mathematics. Specifically, our methodology serves as a robust diagnostic tool to address key open questions in empirical finance, beginning with the validation of the rough volatility paradigm. By precisely testing the roughness parameter $H$, this approach enhances the modeling of asset price fluctuations, which is crucial for modern option pricing and risk management. More importantly, we contextualize the verification of the weak-form Efficient Market Hypothesis (EMH). In this context, $H=1/2$ represents the Brownian benchmark associated with the absence of statistically detectable persistence or anti-persistence in price increments. Confidence intervals lying entirely above or below $1/2$ are interpreted as evidence of positive or negative inefficiency \citep{BianchiPianese2018,carbone2004time}, respectively, while intervals containing $1/2$ are classified as statistically neutral.
\end{itemize}

The remainder of this paper is structured as follows. Section \ref{sec:Pre_Prob} reviews the mathematical preliminaries of fractional Brownian motion and details the convergence problem inherent to the long-memory KS test. Section \ref{sec:GL_derivative} formalizes the discrete Gr\"{u}nwald-Letnikov derivative filter, detailing its action in both direct and spectral spaces and proving how it successfully splits memory from self-similarity. Section \ref{sec:KS_LRD} presents our main theoretical contribution, deriving the limiting distribution of the accelerated GL-KS test statistic. Section \ref{sec:computational_analysis} provides extensive numerical simulations validating the rapid convergence rates and finite-sample properties of the test. Section \ref{sec:Application} presents the empirical application to financial time series, demonstrating the framework's performance on real-world asset data. Finally, Section \ref{sec:Conclusion} concludes.

\section{Preliminaries and problem statement}\label{sec:Pre_Prob}
A \textit{fractional Brownian motion} (fBm) with Hurst exponent $H\in (0,1]$ is a centered Gaussian, non-stationary and self-similar process $B^H_t$ fully characterized by its covariance function
\begin{equation}\nonumber
    \mathbb E\left[B_t^HB_s^H\right] = \frac{1}{2}\left(t^{2H}+s^{2H}-|t-s|^{2H}\right).
\end{equation}

    By definition, fBm is an $H$-self-similar process ($H$-ss), which implies the following scaling invariance in the sense of finite-dimensional distributions
    \begin{equation}\nonumber
        \left\{B^H_{at}\right\}_{t\geq 0}\overset{f.d.d.}{=} a^H\left\{B^H_{t}\right\}_{t\geq 0}, \qquad \forall a>0.
    \end{equation}
A core feature of fBm is that  it is the unique mean-zero Gaussian process with stationary and self-similar increments. Its discrete-time increment process, known as \textit{fractional Gaussian noise} (fGn) $Z_{t,a} = B_{t+a}^H-B^H_{t}$, inherits this underlying scaling structure and serves as the standard statistical foundation for inference on fractional time series. Its covariance function is
\begin{equation}\nonumber
    \begin{array}{ccl}
        K_{H}(k)=K_{Z_{\cdot,a}^H}(t-s) = \frac{1}{2}\left[|k+a|^{2H}+|k-a|^{2H}-2|k|^{2H}\right],\quad t,s\geq 0.
    \end{array}
\end{equation}
From a spectral point of view, the covariance function can be represented as $K_H(k) = \int_{-\pi}^\pi e^{ik\lambda}f_H(\lambda)d\lambda$, where the spectral density is
\begin{equation}\label{eq:spectrum_fGn}
\begin{array}{ccl}
    f_H(\lambda) & = &  \frac{1}{2\pi}\sum\limits_{k\in\mathbb Z}K_H(k)e^{-ik\lambda}\\
     & = & \frac{\sin(\pi H)\Gamma(2H+1)}{\pi}\left(1-\cos(\lambda)\right)\sum\limits_{m\in\mathbb Z}\vert \lambda+2\pi m\vert^{-2H-1}, \quad \lambda\in[-\pi,\pi].
\end{array}
\end{equation}
Writing $C_H=\frac{\sin(\pi H)\Gamma(2H+1)}{\pi}$, near zero the spectral density boils down to
\begin{equation}\label{eq:f_H_near_zero}
    f_H(\lambda)\sim C_H\vert \lambda\vert^{1-2H},\quad \lambda\to 0.
\end{equation}

The fGn inherits the property of self-similarity from the fBm. In particular, a fGn is $H-ss$ in the sense of the finite-dimensional distributions
\begin{equation}\nonumber
    \left\{Z_{at,a}^H\right\}_{t\geq 0}:= \left\{B^H_{a(t+1)}-B^H_{at}\right\}_{t\geq 0}\overset{f.d.d.}{=}a^H\left\{B^H_{t+1}-B^H_t\right\}_{t\geq 0} =: a^H\left\{Z_{t,1}^H\right\}_{t\geq 0}.
\end{equation}

From an empirical point of view, testing the previous self-similarity property of a finite fGn process is expensive in terms of the number of elements. In fact, simulating a fGn process of length $N$, the scaled fGn $Z_{at,a}^H$ has only $N/a$ elements. To overcome this issue, we give the following definition.

\begin{definition}\label{dfn:crossed_fGn}
Let $\left\{Z_{t,a}^H\right\}_{t\geq 0}$ be a fGn with Hurst exponent $H\in(0,1]$ and integer scaling parameter $a\geq 1$. A $crossed$ fGn is defined as \[G_{a,r}^H(t):= Z^H_{r+at,a} = B_{r+a(t+1)}^H-B^H_{r+at}, \quad r=0,\ldots,a-1.\] 
The union $\mathcal Y_{a}^H(t)=\bigcup\limits_{r=0}^{a-1}G_{a,r}^H(t)$ represents the set of all fBm increments with lag $a$.
\end{definition}
Clearly, fixing a branch $r$, a crossed fGn $G_{a,r}^H(t)$ is an $H$-ss process:
\[
\left\{G_{a,r}^H(t)\right\}_{t\geq 0}\overset{f.d.d.}{=}\left\{B^H_{a(t+1)}-B_{at}^H\right\}_{t\geq 0} \overset{f.d.d.}{=}a^H\left\{B_{t+1}^H-B_t^H\right\}_{t\geq 0}=:a^H\left\{G^H_{1,0}(t)\right\}_{t\geq 0}.
\]

By the previous self-similarity property, the covariance function can be written as
\begin{equation}\nonumber
        K_{G_{a,r}^H}(t-s) = \mathbb E\left[ Z^H_{r+at,a} Z^H_{r+as,a}\right] =  \mathbb E\left[Z_{at,a}Z_{as,a}\right]
         = a^{2H} \mathbb E\left[Z_{t,1}^H Z_{s,1}^H\right] = a^{2H}K_{G^H_{1,0}}(t-s).
\end{equation}
Therefore, the relation between the spectral density of an fGn $Z_{t,a}^H$ and a r-th crossed fGn $G_{a,r}^H(t)$ is
\begin{equation}\nonumber
f_{G_{a,r}^H}(\lambda) = a^{2H}f_{H}(\lambda).
\end{equation}

\subsection{Problem Statement and the Long-Memory Issue}
Let \(B^{H_0}=\{B_t^{H_0}\}_{t\ge 0}\) be a fractional Brownian motion with
unknown Hurst parameter \(H_0\in(0,1)\) defined on a probability space $(\Omega, \mathcal{F}, \mathbb{P})$. Since \(B^{H_0}\) has stationary and \(H_0\)-self-similar
increments, its increment process satisfies, for every integer scale \(a\ge 1\),
\[
\big\{B_{a(t+1)}^{H_0}-B_{at}^{H_0}\big\}_{t\ge0}
\overset{f.d.d.}{=}
a^{H_0}\big\{B_{t+1}^{H_0}-B_t^{H_0}\big\}_{t\ge0}.
\]
To exploit this property without losing observations at large scales, we work
with the crossed fractional Gaussian noises
\[
G_{a,r}^{H_0}(t)
:=
B_{r+a(t+1)}^{H_0}-B_{r+at}^{H_0},
\qquad
r=0,\ldots,a-1 ,
\]
where \(t\) ranges over all indices such that the increment is contained in the
observed path. For each fixed branch \(r\),
\[
\{G_{a,r}^{H_0}(t)\}_{t\ge0}
\overset{f.d.d.}{=}
a^{H_0}\{G_{1,0}^{H_0}(t)\}_{t\ge0}.
\]
Thus the unit-scale reference sample is
\begin{equation}\label{eq:X_i_def}
X_i:=G_{1,0}^{H_0}(i)=B_{i+1}^{H_0}-B_i^{H_0},
\qquad i=0,\ldots,n,
\end{equation}
where $n=N-1$, while, for a candidate value $\theta\in(0,1)$, the rescaled
crossed sample at scale $a>1$ is
\begin{equation}\label{eq:Y_tr_def}
Y_{t,r}^{(\theta)}
:=
a^{-\theta}G_{a,r}^{H_0}(t)
=
a^{-\theta}\big(B_{r+a(t+1)}^{H_0}-B_{r+at}^{H_0}\big).
\end{equation}
Pooling all branches gives $\mathcal Y_a^{(\theta)}=\{Y_{t,r}^{(\theta)}:\ r=0,\ldots,a-1,\ t=0,\ldots,m_r-1\}$,
 $m=\sum_{r=0}^{a-1}m_r\simeq n$.

At the theoretical level, if \(\theta=H_0\), then every rescaled crossed branch has
the same finite-dimensional distributions as the unit-scale fGn:
\[
\{Y_{t,r}^{(H_0)}\}_{t\ge0}
\overset{f.d.d.}{=}
\{X_i\}_{i\ge0}.
\]
Equivalently, denoting by \(\Phi\) the distribution function of \(X_i\), one has
\[
Y_{t,r}^{(\theta)}
\overset{d}{=}
a^{H_0-\theta}X_i,
\qquad
\Phi_{a,\theta}(x)
:=
\mathbb P(Y_{t,r}^{(\theta)}\le x)
=
\Phi\big(a^{\theta-H_0}x\big).
\]

For a set of scales \(\mathcal A=[\underline{a},\overline{a}]\subset\mathbb R_+\), define the family of rescaled
distribution functions
\[
\Psi_\theta
:=
\left\{
\Phi_{a,\theta}(x):\ a\in \mathcal A,\ x\in\mathbb R
\right\}
\]
and its Kolmogorov-type diameter
\(
\delta(\theta)
:=
\sup\limits_{x\in\mathbb R}\sup\limits_{a,b\in \mathcal A}
\left|
\Phi\big(a^{\theta-H_0}x\big)
-
\Phi\big(b^{\theta-H_0}x\big)
\right|.
\)

By Propositions 1-3 proved in \citep{Bianchi2004}, the diameter is minimized at the true
self-similarity exponent: it is non-increasing for \(\theta\le H_0\), non-decreasing
for \(\theta\ge H_0\), and it increases away from zero as the scale interval is
enlarged. Hence, in the population case,
\begin{equation}\nonumber
H_0=\argmin_{\theta\in(0,1)}\delta(\theta).
\end{equation}

Empirically, fixing $\underline{a}=1$ and $\overline{a}=a$ for simplicity, the unknown distribution functions are replaced by the empirical
distribution functions
\begin{equation}\nonumber
F_n(x) :=\frac1n\sum_{i=0}^{n-1}\ind_{\{X_i\le x\}},
\qquad
G_m^{(\theta)}(x):=\frac1m\sum_{r=0}^{a-1}\sum_{t=0}^{m_r-1}
\ind_{\{Y_{t,r}^{(\theta)}\le x\}} .
\end{equation}
Therefore, the empirical diameter becomes the two-sample
Kolmogorov--Smirnov statistic
\begin{equation}\nonumber
\widehat\delta_a(\theta)
:=
D_{n,m}(\theta)
=
\sup_{x\in\mathbb R}
\left|
F_n(x)-G_m^{(\theta)}(x)
\right|,
\end{equation}
and the Hurst exponent is estimated by
\begin{equation}\label{eq:Est_Emp_H}
\widehat H
=
\argmin_{\theta\in(0,1)}
\widehat\delta_a(\theta).
\end{equation}
More generally, when several scales are used, one may minimize
\[
\widehat\delta_\mathcal A(\theta)
:=
\sup_{x\in\mathbb R}\sup_{a,b\in \mathcal A}
\left|
\widehat\Phi_a^{(\theta)}(x)-\widehat\Phi_b^{(\theta)}(x)
\right|,
\qquad
\widehat H
=
\argmin_{\theta\in(0,1)}\widehat\delta_\mathcal A(\theta).
\]

The main difficulty is that the two empirical samples are not independent.
They are extracted from the same sample path $B^{H_0}$ and therefore combine two sources
of dependence: the cross-dependence induced by overlapping or nearby branches, and the internal temporal dependence governed by $H$. This affects both the asymptotic distribution of $D_{n,m}$ and the speed at which its finite-sample law approaches the limit.

The asymptotic behavior of the statistic splits into two regimes. In the present paper, these two regimes are used as a benchmark: they identify the obstruction that the Gr\"{u}nwald--Letnikov derivative is designed to remove in Section \ref{sec:GL_derivative}.

Throughout this section we consider only the balanced asymptotic regime $n,m\to\infty,
\frac{n}{m}\to 1,$ or, equivalently, $\frac{n}{n+m}\to \frac12.$ This is the regime corresponding to the empirical construction used in the
applications, where the two samples have the same asymptotic size.

\begin{proposition}[Short-memory benchmark]\label{prop:SRD}
Assume that the auto-covariance and cross-covariance sequences associated with
the two standardized increment samples are absolutely summable. Assume also that $n,m\to\infty,
\frac{n}{m}\to 1.$ Then
\begin{equation}\label{eq:SRD_distribution}
D^{\star}_{n,m}
:=
\sqrt{\frac{nm}{n+m}}D_{n,m}
\Rightarrow
\sup_{x\in\mathbb R}|U(x)|,
\end{equation}
where \(U\) is a centered Gaussian process with covariance kernel
\begin{equation}\label{eq:cov_SRD}
    \operatorname{Cov}(U(x),U(y))=
\frac12\Gamma_X(x,y) +\frac12\Gamma_Y(x,y)
-\frac12\left[\Gamma_{XY}(x,y)+\Gamma_{YX}(x,y)
\right].
\end{equation}
Here \(\Gamma_X\) and \(\Gamma_Y\) are the long-run covariance kernels of the
two marginal empirical processes, while \(\Gamma_{XY}\) and \(\Gamma_{YX}\) are
the corresponding cross-covariance kernels. In particular, $D_{n,m}=O_P\left((n\wedge m)^{-1/2}\right)$.
\end{proposition}

\begin{proof}
See Appendix \ref{app_A}.
\end{proof}
In the short-memory regime, the KS statistic therefore retains
the usual two-sample order. The limiting distribution is not the classical
KS law, because the two samples are extracted from the same
trajectory and are therefore cross-dependent, but the normalization remains the
standard one.
\begin{proposition}[Long-memory benchmark]\label{prop:LRD}
Assume \(H>1/2\). Assume also that $n,m\to\infty,\frac{n}{m}\to 1.$ Then
\begin{equation}\label{eq:LRD_distribution_unfiltered}
D^{\diamond}_{n,m}
:=
\frac{n^{1-H}m^{1-H}}
{n^{1-H}+m^{1-H}}
D_{n,m}
\Rightarrow
\frac{1}{\sqrt{2\pi}}|Z_0|,
\end{equation}
where $Z_0=\frac12\left(Z_X-Z_Y\right).$

Here \((Z_X,Z_Y)\) is the centered Gaussian limit of the normalized linear
partial sums
\begin{equation}\nonumber
n^{-H}\sum_{i=1}^n X_i,
\qquad
m^{-H}\sum_{j=1}^m Y_j.
\end{equation}
Consequently, $Z_0$ is centered Gaussian with variance $\sigma_0^2=\frac14
\left(\sigma_X^2+\sigma_Y^2-2\sigma_{XY}\right),$ where
\begin{equation}\nonumber
    \sigma_X^2
=
\lim_{n\to\infty}
n^{-2H}
\operatorname{Var}
\left(
\sum_{i=1}^n X_i
\right),\qquad \sigma_Y^2
=
\lim_{m\to\infty}
m^{-2H}
\operatorname{Var}
\left(
\sum_{j=1}^m Y_j
\right)\quad 
\end{equation}
\begin{equation}\nonumber
\text{and}\quad \sigma_{XY}
=
\lim_{n,m\to\infty}
(nm)^{-H}
\operatorname{Cov}
\left(
\sum_{i=1}^n X_i,
\sum_{j=1}^m Y_j
\right).
\end{equation}
\end{proposition}

\begin{proof}
See Appendix \ref{app_B}.
\end{proof}

\begin{remark}\label{rmk:slowLRD}
The long-memory case is substantially more delicate from a statistical point of
view. Although the limiting statistic is simpler, since the functional limit
collapses to the deterministic profile \(\phi(x)\) multiplied by a single
Gaussian random variable, the convergence to this limit may be slow.

The reason is that the centered indicator admits the Gaussian projection
decomposition
\begin{equation}\nonumber
\ind_{\{Z\le x\}}-\Phi(x)
=
-\phi(x)Z+r_x(Z),
\qquad
\mathbb E[r_x(Z)Z]=0,
\end{equation}
where $Z\sim \mathcal N(0,1)$. The first term produces the limiting profile
$\phi(x)$, while the residual term $r_x(Z)$ is asymptotically negligible
only at a rate depending on the strength of the long-range dependence.

More precisely, after removing the leading linear projection, the residual
empirical process $R_{n,m}$ satisfies, for $n\asymp m$,
\begin{equation}\nonumber
\left\|
\sup_{x\in\mathbb R}|R_{n,m}(x)|
\right\|_{L^2}
=
\begin{cases}
\mathcal{O}\left((n\wedge m)^{1/2-H}\right),
& 1/2<H<3/4, \\[4pt]
\mathcal{O}\left((n\wedge m)^{-1/4}\sqrt{\log(n\wedge m)}\right),
& H=3/4, \\[4pt]
\mathcal{O}\left((n\wedge m)^{H-1}\right),
& 3/4<H<1.
\end{cases}
\end{equation}
Thus, the closer $H$ is to one, the slower the residual component vanishes.
This explains why the KS statistic becomes
progressively less stable in the long-memory regime: finite-sample deviations
from the asymptotic law persist for much longer samples.

The Gr\"{u}nwald--Letnikov transformation introduced below is designed precisely
to remove this long-memory obstruction before applying the
Kolmogorov--Smirnov comparison. By weakening the low-frequency persistence of
the increment process, the filtered statistic is brought back to a
short-memory empirical-process regime, where the convergence to the limiting
law is substantially more stable.
\end{remark}
\begin{remark} At first sight, since fBm is Gaussian, the study of its first two moments would be sufficient to characterize its finite-dimensional distributions. In this sense, for a purely Gaussian fBm, a covariance-based approach would already contain the full probabilistic information.

\begin{table}[!ht]
\centering

\caption{\label{tab:acf-decay} Long-lag dependence in representative fractional processes.}
\small
\begin{tabular}{>{\raggedright}m{3.2cm} >{\centering}m{5cm} m{3.5cm}}
\toprule
\textbf{Process} & \textbf{Asymptotic behavior} & \textbf{Key parameter} \\
\midrule

fGn &
$\mathbb{E}[Z_{t}Z_{t+\tau}] \sim |\tau|^{2H-2}$ & 
 \( H \in (0,1) \)\\
 \vspace{.15cm}
fPP &$\mathbb{E}[N(t)N(t+\tau)] \sim |\tau|^{-\alpha}$ & 
 \( \alpha \in (0,1) \) \\
 \vspace{.15cm}
 fLm & 
$\mathbb{E}[L_H(t)L_H(t+\tau)] \sim |\tau|^{2H-2}$ & 
 \( H \in (0,1) \) \\
  \vspace{.15cm}
 ARFIMA & 
$\rho(\tau) \sim |\tau|^{2d-1}$ & 
 \( d \in (0, 1/2) \) \\
 \vspace{.15cm}
fOU, Langevin type &$\gamma(\tau) \sim C_{H,\lambda} |\tau|^{2H-2}$
&$H\in(0,1)$ \\
fOU, Lamperti type &$\gamma(\tau) \sim C e^{-\lambda \tau}$&$H\in(0,1)$ \\
\bottomrule
\end{tabular}
\vspace{0.15cm} \begin{minipage}{0.94\textwidth} \footnotesize \emph{Notes.} fGn denotes fractional Gaussian noise; fPP fractional Poisson process; fLm fractional L\'evy motion; ARFIMA autoregressive fractionally integrated moving average; fOU fractional Ornstein--Uhlenbeck process. The table is intended only as a qualitative benchmark: the listed processes are not assumed to be distributionally equivalent to fGn. In infinite-variance L\'evy-type settings, covariance-based expressions must be replaced by appropriate dependence notions such as codifference or covariation. \end{minipage}
\end{table}

The reason for adopting a distributional KS criterion is different. The aim is not only to recover the covariance structure of Gaussian fBm, but to build a non-parametric self-similarity diagnostic that can be transported, at least in principle, to broader classes of fractional processes. In such settings, the first two moments may be insufficient, unstable, or even undefined. A distributional comparison of rescaled increments is therefore more flexible: it tests the scaling relation directly at the level of empirical distributions rather than only through moment scaling.
\end{remark}
The fGn is used here as the canonical benchmark because it is the stationary increment process of fBm, it has an explicit covariance and spectral structure, and it exhibits the standard power-law memory pattern governed by the Hurst exponent. Moreover, many fractional models used in applications display analogous hyperbolic memory, or can be compared with fGn at the level of their long-lag dependence. Table~\ref{tab:acf-decay} summarizes the relevant asymptotic regimes. The purpose of the table is not to claim that all these models are identical to fGn, but to justify why fGn provides the natural reference model for deriving the asymptotic theory developed below.

\section{Gr\"{u}nwald-Letnikov-Kolmogorov--Smirnov test}\label{sec:GL_derivative}
To overcome the slow empirical convergence rates caused by LRD inside the classical KS test framework, we introduce in Section \ref{sec:GL_filter} a mathematical architecture designed to accelerate convergence by decoupling the long-memory properties of the process from its underlying self-similar scaling behavior. This is achieved by filtering the fGn through a discrete fractional transformation based on the Gr\"{u}nwald-Letnikov (GL) derivative. In Section \ref{sec:KS_LRD} we will introduce the KS asymptotic distribution for LRD after a GL filtration.

\subsection{Gr\"{u}nwald-Letnikov filter}\label{sec:GL_filter}
\begin{definition}[Discrete Gr\"{u}nwald-Letnikov derivative \citep{OuannasEtal2023}]
The discrete Gr\"{u}nwald-Letnikov derivative of order $\alpha \in \mathbb{R}^+$ with step size $h > 0$ for a real-valued function $f(x) \in \mathbb{R}$ is defined as:

\begin{equation}\label{eq:GL_definition}
    \Delta_{h}^{GL,\alpha}f(x) = \frac{1}{h^{\alpha}}\sum_{k=0}^{\infty}(-1)^{k}\binom{\alpha}{k}f(x-kh) = \frac{1}{h^{\alpha}}\sum_{k=0}^{\infty}\omega_{k}(\alpha)f(x-kh),
\end{equation}
    where the GL binomial coefficients $\omega_{k}(\alpha)=\frac{\Gamma(k-\alpha)}{\Gamma(-\alpha)\Gamma(k+1)}$.
\end{definition}

In terms of the lag operator $L_hf(x):=f(x-h)$, and the repeated lag operator $L_h^kf(x):=f(x-hk)$, the discrete GL derivative operator can be compactly expressed as a fractional power of the differencing operator:

\begin{equation}\label{eq:GL_lag_operator}
    \Delta_{h}^{GL,\alpha}f(x) = \frac{1}{h^{\alpha}}\sum_{k=0}^{\infty}\omega_{k}(\alpha)L_{h}^{k}f(x) = h^{-\alpha}(1-L_{h})^{\alpha}f(x).
\end{equation}
\subsubsection{Filter application at the direct space}
We now explore how the GL filter acts upon fractional Brownian motion and its increments in the time domain, focusing specifically on how it modifies the self-similarity index depending on the choice of the scaling operator definition.

\begin{proposition}\label{prop:fBm_GL_scaling}
Let $B_{t}^{H}$ be a fBm with Hurst exponent $H\in(0,1)$, and let $a\in \mathbb R^+$ be a scale parameter. Then, the discrete GL derivative applied to the fBm process is $H-ss$:
\begin{equation}\label{eq:prop33_1}
    \Delta_{h}^{GL,\alpha}B_{at}^{H} \stackrel{f.d.d.}{=} a^{H} \cdot \Delta_{h}^{GL,\alpha}B_{t}^{H}.
\end{equation}
\end{proposition}
\begin{proof}
To establish \eqref{eq:prop33_1}, we expand the definition of the derivative operator $\Delta_{h}^{GL,\alpha}$ and utilize the intrinsic $H$-self-similarity of the underlying fBm process ($B_{at}^H \stackrel{f.d.d.}{=} a^H B_t^H$):

\begin{equation}\nonumber
    \Delta_{h}^{GL,\alpha}B_{at}^{H} = \frac{1}{h^{\alpha}}\sum_{k=0}^{\infty}\omega_{k}(\alpha)B_{a(t-kh)}^{H} \stackrel{f.d.d.}{=} a^{H} \left( \frac{1}{h^{\alpha}}\sum_{k=0}^{\infty}\omega_{k}(\alpha)B_{t-kh}^{H} \right) = a^{H} \cdot \Delta_{h}^{GL,\alpha}B_{t}^{H}.
\end{equation}
\end{proof}
The core operational breakthrough of this paper relies on applying this filtering technique to the crossed fGn framework defined in Section \ref{sec:Pre_Prob}. The next proposition guarantees that filtering the multi-scale branches preserves the exact distributional scaling law required for our test statistic in the finite-dimensional sense.

\begin{proposition}\label{prop:crossed_fGn_GL}
The discrete Gr\"{u}nwald-Letnikov derivative $\Delta_{h}^{GL,\alpha}$ of a crossed fGn process $G_{a,r}^H(t)$ remains strictly $H$-self-similar for any scale $a \ge 1$ and branch $r = 0, \dots, a-1$:
\begin{equation}\nonumber
\Delta_{h}^{GL,\alpha}G_{a,r}^{H}(t) \stackrel{f.d.d.}{=} a^{H} \cdot \Delta_{h}^{GL,\alpha}G_{1,0}^{H}(t).
\end{equation}
\end{proposition}
\begin{proof}
By mapping the definition of the crossed fGn $G_{a,r}^H(t) = B_{r+a(t+1)}^H - B_{r+at}^H$ into the GL operator, we obtain:
\begin{equation}\nonumber
    \Delta_{h}^{GL,\alpha}G_{a,r}^{H}(t) = \frac{1}{h^{\alpha}}\sum_{k=0}^{\infty}\omega_{k}(\alpha)\left( B_{r+a(t-kh+1)}^{H} - B_{r+a(t-kh)}^{H} \right).
\end{equation}
Since fBm features strictly stationary increments, the deterministic time shift $r$ can be dropped without modifying the joint finite-dimensional distributions of the process:
\begin{equation}\nonumber
    \Delta_{h}^{GL,\alpha}G_{a,r}^{H}(t) \stackrel{f.d.d.}{=} \frac{1}{h^{\alpha}}\sum_{k=0}^{\infty}\omega_{k}(\alpha)\left( B_{a(t-kh+1)}^{H} - B_{a(t-kh)}^{H} \right).
\end{equation}
Finally, factoring out the scale parameter $a$ using the $H$-self-similarity of fBm yields the targeted relation:
\begin{equation}\nonumber
    \begin{array}{ccl}
        \Delta_{h}^{GL,\alpha}G_{a,r}^{H}(t) &\stackrel{f.d.d.}{=}& a^{H} \left[ \frac{1}{h^{\alpha}}\sum_{k=0}^{\infty}\omega_{k}(\alpha)\left( B_{t-kh+1}^{H} - B_{t-kh}^{H} \right) \right] \\
        &=& a^{H} \cdot \Delta_{h}^{GL,\alpha}G_{1,0}^{H}(t).
    \end{array}
\end{equation}
This completes the proof.
\end{proof}

\begin{corollary}[Identification invariance under GL filtering]\label{cor:invariance}
Let $B^{H_0}$ be an fBm with $H_0\in(0,1)$, and let
$T_\alpha=\sigma_{\alpha,h}^{-1}\Delta^{GL,\alpha}_h$ be the standardized GL filter, with $\alpha\ge 0$. For a fixed scale $a>1$, define the population filtered KS criterion
\[D_\alpha(\theta)
=
\sup_{x\in\mathbb R}
\left|
\Phi_\alpha(x)
-
\Phi_\alpha\left(a^{\theta-H_0}x\right)
\right|,
\]
where $\Phi_\alpha$ is the distribution function of
$T_\alpha G^{H_0}_{1,0}$. Then
\[
D_\alpha(\theta)=0
\quad\Longleftrightarrow\quad
\theta=H_0.
\]
Therefore, the GL transformation preserves the population identification of the Hurst exponent. It may change the empirical criterion and its asymptotic covariance structure, but it does not change the self-similarity
parameter identified by the minimum-distance problem.
\end{corollary}

\begin{proof}
    By Proposition \ref{prop:crossed_fGn_GL},
\[
T_\alpha G^{H_0}_{a,r}
\stackrel{f.d.d.}{=}
a^{H_0}T_\alpha G^{H_0}_{1,0}.
\]
After rescaling by $a^{-\theta}$, the filtered multi-scale sample has distribution $a^{H_0-\theta}T_\alpha G^{H_0}_{1,0}.$ Hence, its distribution coincides with that of the unit-scale filtered sample if and only if $a^{H_0-\theta}=1$. Since $a>1$, this is equivalent
to $\theta=H_0$. Therefore, the population KS diameter is uniquely minimized at $H_0$.
\end{proof}

\subsubsection{Filter application at the spectral space}
While Proposition \ref{prop:crossed_fGn_GL} establishes that the GL filter leaves the structural scaling multiplier $a^H$ perfectly intact across scales, we now demonstrate its effect in the frequency domain, proving that it simultaneously eliminates the long-memory singularity at the origin.

Let $Y_{a,r}^{H,\alpha}(t) := \Delta_{h}^{GL,\alpha}G_{a,r}^{H}(t)$ denote the GL-filtered crossed fGn process. Its spectral characterization is formalized below.

\begin{proposition}\label{prop:spectral_GL}
The spectral density of the discrete GL-filtered crossed fGn process $Y_{a,r}^{H,\alpha}(t)$ is given by:
\begin{equation}\label{eq:filtered_spectrum}
    f_{Y}^{(h)}(\lambda) = a^{2H} h^{-2\alpha} \left| 1 - e^{-ih\lambda} \right|^{2\alpha} f_{H}(\lambda), \quad \lambda \in [-\pi, \pi],
\end{equation}
where $f_H(\lambda)$ represents the spectral density of the unscaled fGn process (Equation \eqref{eq:spectrum_fGn}).
\end{proposition}

\begin{proof}
The discrete GL derivative operator $\Delta_{h}^{GL,\alpha} = h^{-\alpha}(1-L_h)^\alpha$ represents a linear filter whose characteristic transfer function in the frequency domain is given by $A_{\alpha,h}(\lambda) = h^{-\alpha}(1-e^{-ih\lambda})^\alpha$. Under linear filtering theory, the spectral density of the output process is equal to the spectral density of the input process multiplied by the squared modulus of the transfer function:
\begin{equation}\nonumber
    f_{Y}^{(h)}(\lambda) = \left| A_{\alpha,h}(\lambda) \right|^{2} f_{G_{a,r}^{H}}(\lambda) = \frac{1}{h^{2\alpha}} \left| 1 - e^{-ih\lambda} \right|^{2\alpha} f_{G_{a,r}^{H}}(\lambda).
\end{equation}
Substituting the crossed fGn spectral relation $f_{G_{a,r}^{H}}(\lambda) = a^{2H}f_{H}(\lambda)$ into the filter formula directly yields \eqref{eq:filtered_spectrum}, completing the proof.
\end{proof}
\begin{proposition}\label{prop:asymptotic_spectral_GL}
In the low-frequency limit $\lambda \to 0$, the spectral density of the filtered process $Y_{a,r}^{H,\alpha}(t)$ behaves as:
\begin{equation}\label{eq:asymptotic_spectral_shift}
    f_{Y}^{(h)}(\lambda) \sim a^{2H} C_{H} |\lambda|^{1-2(H-\alpha)}, \quad \text{as } \lambda \to 0,
\end{equation}
where $C_H = \frac{\sin(\pi H)\Gamma(2H+1)}{\pi}$.
\end{proposition}

\begin{proof}
We examine the asymptotic behavior of the individual terms in Equation \eqref{eq:filtered_spectrum} as $\lambda \to 0$. First, applying a first-order Taylor expansion to the exponential component of the transfer function yields:
\begin{equation}\nonumber
    \left| 1 - e^{-ih\lambda} \right|^{2\alpha} = \left| 1 - \left(1 - ih\lambda + \mathcal{O}(\lambda^2)\right) \right|^{2\alpha} \sim |ih\lambda|^{2\alpha} = h^{2\alpha}|\lambda|^{2\alpha}.
\end{equation}
Second, from Equation \eqref{eq:f_H_near_zero}, we know that the unscaled fGn spectral density scales near the origin as $f_H(\lambda) \sim C_H|\lambda|^{1-2H}$. Combining these limits inside the expression for $f_Y^{(h)}(\lambda)$ leads to:
\begin{equation}\nonumber
    f_{Y}^{(h)}(\lambda) \sim \left( \frac{1}{h^{2\alpha}} \right) \left( h^{2\alpha}|\lambda|^{2\alpha} \right) \left( a^{2H} C_{H} |\lambda|^{1-2H} \right) = a^{2H} C_{H} |\lambda|^{2\alpha + 1 - 2H}.
\end{equation}
Gathering the exponents of $|\lambda|$ simplifies the expression directly to $a^{2H} C_{H} |\lambda|^{1-2(H-\alpha)}$, completing the proof.
\end{proof}

By choosing a fractional filter order $\alpha$ such that $H - \alpha < 1/2$, the exponent $1 - 2(H-\alpha)$ in Equation \eqref{eq:asymptotic_spectral_shift} becomes strictly positive. Consequently, the spectral density no longer diverges at the origin $\lim_{\lambda \to 0} f_Y^{(h)}(\lambda) = 0$, meaning the long-memory singularity is successfully removed. The filtering operation shifts the process out of the LRD domain and into a short-memory or anti-persistent domain while preserving the scaling constant $a^{2H}$ intact. This dual feature provides the mathematical framework required to derive stable, fast-converging non-parametric test statistics under short-range dependence functional limit theorems, as formalized in Section \ref{sec:KS_LRD}.

\subsection{KS asymptotic distribution for LRD}\label{sec:KS_LRD}

Having established that the discrete Gr\"{u}nwald-Letnikov filter maps a long-memory fractional Gaussian noise into a SRD regime while preserving its exact finite-dimensional scaling laws, we are now positioned to derive the asymptotic distribution of the accelerated test statistic.

Before introducing the filtered samples, we specify the null hypothesis considered in this subsection. For a fixed \(H_0\in(1/2,1)\), the null hypothesis is the correctly specified self-similarity relation \[ \mathcal H_0(H_0):\qquad \left\{ a^{-H_0}G^{H_0}_{a,r}(t) \right\}_{t\geq0} \overset{f.d.d.}{=} \left\{ G^{H_0}_{1,0}(t) \right\}_{t\geq0}, \qquad r=0,\ldots,a-1. \] This is not the classical independence null of the two-sample KS test, since both samples are extracted from the same trajectory. The aim is to derive the limiting law of the KS distance under this dependent self-similarity null.

Let $\{X_i\}_{i=1}^n$ and $\{Y_{t,r}\}$ ($t=0,\dots,m_r-1$; $r=0,\dots,a-1$) be the unit-lag and pooled multi-scale samples defined in Equations \eqref{eq:X_i_def} and \eqref{eq:Y_tr_def} with $\theta=H_0$, respectively. We apply the pre-scaled discrete GL derivative operator $\Delta_{h}^{GL,\alpha}$ of order $\alpha > H_0 - 1/2$ to both samples. Under \(\mathcal H_0(H_0)\), we normalize the filtered increments by their common theoretical standard deviation $\sigma_{\alpha,h} := \left(\mathbb{E}\left[(\Delta_{h}^{GL,\alpha}X_1)^2\right]\right)^{1/2}$. This yields the standardized, GL-filtered stochastically equivalent sequences:
\begin{equation}\nonumber
    \tilde{X}_i = \frac{1}{\sigma_{\alpha,h}} \Delta_{h}^{GL,\alpha}X_i, \quad \tilde{Y}_{t,r} = \frac{1}{\sigma_{\alpha,h}} \Delta_{h}^{GL,\alpha}Y_{t,r}.
\end{equation}
The pooled multi-scale filtered sample $\{\tilde{Y}_j\}_{j=1}^m$ is constructed by aggregating all branches $\tilde{Y}_{t,r}$, maintaining a total sample size of $m = \sum_{r=0}^{a-1} m_r \approx n$. Under \(\mathcal H_0(H_0)\), Proposition \ref{prop:crossed_fGn_GL} guarantees that $\tilde{X}_i \stackrel{d}{=} \tilde{Y}_j \sim \mathcal{N}(0,1)$ marginally, and their joint finite-dimensional distributions satisfy $\{\tilde{Y}_{t,r}\}_{t \ge 0} \stackrel{f.d.d.}{=} \{\tilde{X}_i\}_{i \ge 1}$.

We define the corresponding ECDFs of the filtered samples as:
\begin{equation}\nonumber
    \tilde{F}_n(x) = \frac{1}{n}\sum_{i=1}^{n}\ind_{\{\tilde{X}_i \le x\}}, \quad \tilde{G}_m(x) = \frac{1}{m}\sum_{r=0}^{a-1}\sum_{t=0}^{m_r - 1}\ind_{\{\tilde{Y}_{t,r} \le x\}}.
\end{equation}
The Gr\"{u}nwald-Letnikov-Kolmogorov--Smirnov (GL-KS) test statistic is then formalized as:
\begin{equation}\nonumber
    \tilde{D}_{n,m} = \sup_{x \in \mathbb{R}} |\tilde{F}_n(x) - \tilde{G}_m(x)|.
\end{equation}

The main theoretical contribution of this paper is established in the following theorem, which demonstrates that the GL filter successfully neutralizes the long-memory convergence bottleneck.

\begin{theorem}[Asymptotic Distribution of the GL-KS Statistic]\label{thm:GL_KS_limit}
Let $B_t^{H_0}$ be a fBm with Hurst parameter $H_0 \in (1/2, 1)$. If the GL filter order $\alpha$ is chosen such that $H_0 - \alpha < 1/2$, then as $n, m \to \infty$ with $\frac{n}{n+m} \to \lambda \in (0,1)$, the standardized GL-KS statistic satisfies:
\begin{equation}\label{eq:GL_KS_weak_conv}
    D_{n,m}^\dag := \sqrt{\frac{nm}{n+m}}\tilde{D}_{n,m} \xrightarrow{d} \sup_{x \in \mathbb{R}} |\tilde{U}(x)|,
\end{equation}
where $\widetilde U$ is a centered continuous Gaussian process on $\mathbb R$. Its covariance kernel is 
\[ \operatorname{Cov}\big(\widetilde U(x),\widetilde U(y)\big) = (1-\lambda)\Gamma_{\widetilde X\widetilde X}(x,y) + \lambda\Gamma_{\widetilde Y\widetilde Y}(x,y) - \sqrt{\lambda(1-\lambda)} \left[ \Gamma_{\widetilde X\widetilde Y}(x,y) + \Gamma_{\widetilde Y\widetilde X}(x,y) \right], \]
where, for $A,B\in\{\widetilde X,\widetilde Y\}$, $\Gamma_{AB}(x,y) = \sum_{k\in\mathbb Z} \operatorname{Cov} \left( \mathds 1_{\{A_0\le x\}}, \mathds 1_{\{B_k\le y\}} \right).$

Equivalently, by the Hermite expansion of the centered Gaussian indicator, \[ \Gamma_{AB}(x,y) = \phi(x)\phi(y) \sum_{\ell=1}^{\infty} \frac{He_{\ell-1}(x)He_{\ell-1}(y)}{\ell!} \sum_{k\in\mathbb Z}\rho_{AB}(k)^\ell , \]
where $He_{\ell}(\cdot)$ denotes the probabilist's Hermite polynomials of order $\ell$, $\phi(\cdot)$ is the standard Gaussian PDF, and $\rho_{\tilde{X}}(k)$, $\rho_{\tilde{Y}}(k)$, $\rho_{\tilde{X}\tilde{Y}}(k)$ are the autocorrelation and cross-correlation functions of the GL-filtered Gaussian sequences.
\end{theorem}

\begin{proof}
Let $\widetilde\beta_n(x)=\sqrt{n}\big(\widetilde F_n(x)-\Phi(x)\big)$, and $\widetilde\gamma_m(x)=\sqrt{m}\big(\widetilde G_m(x)-\Phi(x)\big)$. We first establish the joint weak convergence of $\big(\widetilde\beta_n,\widetilde\gamma_m\big)$ in $\ell^\infty(\mathbb R)\times\ell^\infty(\mathbb R)$.

For any standardized Gaussian random variable \(\xi\), the centered indicator admits the Hermite
expansion
\begin{equation}\label{eq:indicator}
\mathds 1_{\{\xi\le x\}}-\Phi(x)
=
-\phi(x)\sum_{\ell=1}^{\infty}
\frac{He_{\ell-1}(x)}{\ell!}He_\ell(\xi).
\end{equation}
The Hermite rank is equal to one, because the first coefficient is non-zero.

In the unfiltered long-memory case, the autocorrelation behaves as $\rho(k)\sim C|k|^{2H_0-2}$, and for $H_0>1/2$ the series $\sum_k |\rho(k)|$ is not summable. The Hermite reduction principle for long-range dependent Gaussian sequences then implies that the first-order Hermite projection dominates the empirical process, producing the non-standard normalization and the collapsed limit described in Proposition~\ref{prop:LRD}; see \cite{taqqu1975weak,dobrushin1979non,dehling1989}.

After the GL filtering, Proposition~\ref{prop:asymptotic_spectral_GL} gives
\[
f_{\widetilde Y}^{(h)}(\lambda)\sim C|\lambda|^{1-2(H_0-\alpha)},
\qquad \lambda\to0.
\]
Equivalently, the long-lag correlations of the filtered sequences satisfy
\[
\rho_{\widetilde X}(k)\sim C_{\alpha,h}|k|^{2(H_0-\alpha)-2},
\qquad
\rho_{\widetilde Y}(k)\sim C_{\alpha,h}|k|^{2(H_0-\alpha)-2}.
\]
Since \(H_0-\alpha<1/2\), we have $2(H_0-\alpha)-2<-1,$ and therefore
\begin{equation}\label{eq:rhos}
\sum_{k\in\mathbb Z}|\rho_{\widetilde X}(k)|<\infty,
\qquad
\sum_{k\in\mathbb Z}|\rho_{\widetilde Y}(k)|<\infty,
\qquad
\sum_{k\in\mathbb Z}|\rho_{\widetilde X\widetilde Y}(k)|<\infty.
\end{equation}
Moreover, since \(|\rho(k)|\le1\), the same summability holds for every Hermite power
\(\rho(k)^\ell\), \(\ell\ge1\). Hence, each Hermite component in \eqref{eq:indicator} satisfies a standard central limit theorem of Breuer--Major type, and the Hermite series can be handled by the usual empirical-process central limit theorem for short-memory Gaussian subordinated sequences; see \cite{breuer1983central,Arcones1994,dehling1989}. Consequently,
\[
\big(\widetilde\beta_n,\widetilde\gamma_m\big)
\Rightarrow
\big(U_X,U_Y\big)
\qquad\text{in }\ell^\infty(\mathbb R)^2,
\]
where $(U_X,U_Y)$ is a centered bivariate Gaussian process.

For \(A,B\in\{\widetilde X,\widetilde Y\}\), define the long-run covariance kernel
\[
\Gamma_{AB}(x,y)
=
\sum_{k\in\mathbb Z}
\operatorname{Cov}
\left(
\mathds 1_{\{A_0\le x\}},
\mathds 1_{\{B_k\le y\}}
\right).
\]
Equivalently, using the Hermite expansion,
\[
\Gamma_{AB}(x,y)
=
\phi(x)\phi(y)
\sum_{\ell=1}^{\infty}
\frac{He_{\ell-1}(x)He_{\ell-1}(y)}{\ell!}
\sum_{k\in\mathbb Z}\rho_{AB}(k)^\ell .
\]
The absolute summability in \eqref{eq:rhos} guarantees that these kernels are finite.

Since \(n/(n+m)\to\lambda\in(0,1)\), we have
\[
\sqrt{\frac{nm}{n+m}}
\big(\widetilde F_n(x)-\widetilde G_m(x)\big)
=
\sqrt{\frac{m}{n+m}}\,\widetilde\beta_n(x)
-
\sqrt{\frac{n}{n+m}}\,\widetilde\gamma_m(x),
\]
and therefore
\[
\sqrt{\frac{nm}{n+m}}
\big(\widetilde F_n-\widetilde G_m\big)
\Rightarrow
\widetilde U
\qquad\text{in }\ell^\infty(\mathbb R),
\]
where $\widetilde U(x)=\sqrt{1-\lambda}\,U_X(x)-\sqrt{\lambda}\,U_Y(x)$. Its covariance kernel is
\[
\operatorname{Cov}\big(\widetilde U(x),\widetilde U(y)\big)
=
(1-\lambda)\Gamma_{\widetilde X\widetilde X}(x,y)
+\lambda\Gamma_{\widetilde Y\widetilde Y}(x,y)
-\sqrt{\lambda(1-\lambda)}
\left[
\Gamma_{\widetilde X\widetilde Y}(x,y)
+
\Gamma_{\widetilde Y\widetilde X}(x,y)
\right].
\]
In the balanced case \(\lambda=1/2\), this reduces to the same covariance structure as in Proposition~\ref{prop:SRD}, with the unfiltered kernels replaced by their GL-filtered counterparts.

Finally, since the limit process has continuous sample paths and the map
$f\mapsto \sup_{x\in\mathbb R}|f(x)|$ is continuous on the limiting support, the continuous mapping
theorem gives
\begin{equation}\nonumber
D_{n,m}^{\dagger}
=
\sqrt{\frac{nm}{n+m}}\,
\widetilde D_{n,m}
\Rightarrow
\sup_{x\in\mathbb R}|\widetilde U(x)|.
\end{equation}
\end{proof}

\begin{remark}
Theorem \ref{thm:GL_KS_limit} confirms that the GL-KS statistic bypasses the LRD phase transition. While the classical unfiltered statistic requires a non-standard normalization $n^{1-H}$ and delivers a degenerate absolute Gaussian limit for $H>1/2$, the filtered alternative retains a stable $\sqrt{n}$ rate and converges to a rich, functional Gaussian supremum across the entire parameter space. This structural rehabilitation eliminates the inflation of Type I errors and underpins the fast numerical convergence rates analyzed in Section \ref{sec:computational_analysis}.
\end{remark}

\subsubsection{Finite-Sample Truncation and Burn-in Mechanics}\label{subsec:truncation_error}

In practical applications, the underlying fGn sequence is observed only over a finite time horizon. Consequently, the infinite summation defining the discrete Gr\"{u}nwald--Letnikov derivative cannot be implemented exactly. This is the standard initialization problem arising in finite-sample implementations of infinite-order linear filters and fractional-difference operators; see \cite{brockwell2009time,giraitis2012large,podlubny1998fractional}. In the present setting, the issue is particularly relevant because the GL coefficients decay only polynomially. We therefore introduce a burn-in deletion rule whose purpose is to remove the finite-sample effect of the unobserved pre-sample history.

Let $X_i$ be the unscaled fGn process. For each $i=1,\ldots,n$, the theoretical infinite-memory filtered variable $\widetilde X_i$ and its finite-sample truncated counterpart $\widehat X_i$ are related by
\begin{equation}\nonumber
\widehat X_i = \frac{1}{\sigma_{\alpha,h}} \sum_{k=0}^{i-1}\omega_k(\alpha)X_{i-k} = \widetilde X_i-E_i,
\end{equation}
where $E_i = \frac{1}{\sigma_{\alpha,h}} \sum_{k=i}^{\infty}\omega_k(\alpha)X_{i-k}$ is the omitted pre-sample tail. Since the GL coefficients satisfy $\omega_k(\alpha)\sim \frac{k^{-\alpha-1}}{\Gamma(-\alpha)}$, $k\to\infty$, and the fGn autocovariance satisfies $\gamma_X(k)\sim C_H k^{2H-2}$, $k\to\infty,$ the variance of the truncation error satisfies
\begin{equation}\nonumber
\mathbb E[E_i^2] \le C i^{2H-2\alpha-2} = C i^{-2(1+\alpha-H)}.
\end{equation} 
Therefore, after deleting the first $N_0=\lfloor n^\gamma\rfloor$ observations, with 
\begin{equation} \label{eq:gamma_condition}
\gamma>\frac{1}{2(1+\alpha-H)},
\end{equation}
we obtain, uniformly for $i\ge N_0$, 
\begin{equation}\label{eq:error}
\mathbb E[E_i^2] \le C n^{-2\gamma(1+\alpha-H)}.
\end{equation}

\begin{lemma}[Uniform negligibility of the truncation error] \label{lem:burn-in}
Let $B^H$ be a fBm with $H\in(1/2,1)$, and let the GL order satisfy $H-\alpha<1/2$. If the burn-in exponent $\gamma$ satisfies condition \eqref{eq:gamma_condition}, then 
\begin{equation} \nonumber
\sup_{x\in\mathbb R} \sqrt n \left| \widehat F_n(x)-\widetilde F_n(x) \right| \overset{P}{\longrightarrow}0,
\end{equation}
where \(\widehat F_n(x)\) and \(\widetilde F_n(x)\) denote the empirical distribution functions of the finite-sample truncated filtered sequence and of the theoretical infinite-memory filtered sequence, respectively.
\end{lemma}

\begin{proof}
Set $\kappa=\gamma(1+\alpha-H).$ By assumption, \(\kappa>1/2\). From \eqref{eq:error}, \[ \sup_{i\ge N_0}\mathbb E[E_i^2]\le C n^{-2\kappa}, \qquad \sup_{i\ge N_0}\mathbb E|E_i|\le C n^{-\kappa}. \] 
For each $x\in\mathbb R$, 
\[ \left| \mathds 1_{\{\widehat X_i\le x\}} - \mathds 1_{\{\widetilde X_i\le x\}} \right| \le \mathds 1_{\{|\widetilde X_i-x|\le |E_i|\}}. \]
Since \((\widetilde X_i,E_i)\) is jointly Gaussian and \(\operatorname{Var}(E_i)\to0\), the conditional density of \(\widetilde X_i\) given \(E_i\) is uniformly bounded for all sufficiently large \(n\). Hence there exists a constant \(C>0\), independent of \(x\) and \(i\), such that \[ \sup_{x\in\mathbb R} \mathbb P\left(|\widetilde X_i-x|\le |E_i|\right) \le C\mathbb E|E_i| \le C n^{-\kappa}.\] 
Therefore, \[ \mathbb E \left[ \sup_{x\in\mathbb R} \left| \widehat F_n(x)-\widetilde F_n(x) \right| \right] \le \frac{1}{n-N_0} \sum_{i=N_0+1}^n \sup_{x\in\mathbb R} \mathbb P\left(|\widetilde X_i-x|\le |E_i|\right) \le C n^{-\kappa}.  \] Multiplying by \(\sqrt n\) gives 
\[ \sqrt n\, \mathbb E \left[ \sup_{x\in\mathbb R} \left| \widehat F_n(x)-\widetilde F_n(x) \right| \right] \le C n^{1/2-\kappa} \longrightarrow0, \]
because \(\kappa>1/2\). Markov's inequality yields 
\[ \sup_{x\in\mathbb R} \sqrt n \left| \widehat F_n(x)-\widetilde F_n(x) \right| \overset{P}{\longrightarrow}0. \] 
The argument for the pooled crossed sample is identical after applying the same burn-in deletion inside each branch. This proves the claim.
\end{proof}

By virtue of Lemma \ref{lem:burn-in}, the difference between the two statistics is negligible at the $\sqrt n$-scale. Consequently, the asymptotic distribution derived in Theorem \ref{thm:GL_KS_limit} remains invariant under finite-sample truncation.

\begin{remark}[Effective sample sizes in the pooled crossed sample]
\label{rmk:effective_sample_sizes}
The notation used above suppresses a minor but important implementation detail. For the unit-scale sample, the burn-in deletion simply removes the first $N_0=\lfloor n^\gamma\rfloor$
filtered observations, so that the effective sample size is $n_{\mathrm{eff}}:=n-N_0.$
For the multi-scale sample, however, the observations are obtained by pooling the crossed branches
$\widehat Y_{t,r}$, with $t=0,\ldots,m_r-1$, and $r=0,\ldots,a-1$. Since the finite-memory approximation of the GL filter is initialized separately along each branch, the burn-in deletion must also be applied branch by branch. Thus, if
$M_{0,r}=\lfloor m_r^\gamma\rfloor$
denotes the burn-in length of branch $r$, the effective size of the pooled multi-scale sample is
$m_{\mathrm{eff}}
:=
\sum_{r=0}^{a-1}(m_r-M_{0,r})$.
Equivalently, the empirical distribution function computed in finite samples is
$\widehat G_m(x)=\frac{1}{m_{\mathrm{eff}}}
\sum_{r=0}^{a-1}
\sum_{t=M_{0,r}+1}^{m_r}
\ind_{\left\{\widehat Y_{t,r}\le x\right\}}$.
When the branch lengths are asymptotically balanced and the scale $a$ is fixed, one has $m_{\mathrm{eff}}\sim m$,
because $M_{0,r}/m_r\to0$ for every branch. Therefore, this branch-wise correction has no first-order effect on the asymptotic distribution derived above, but it is the relevant finite-sample sample size used in the numerical and empirical implementations.
\end{remark}

\subsection{GL-KS estimator}
The previous section derived the asymptotic distribution of the GL--KS statistic under a fixed null value of the Hurst parameter. We now invert this testing procedure and define an estimator of the unknown self-similarity exponent. Let \(H_0\in(1/2,1)\) denote the true Hurst parameter of the underlying fractional Brownian motion. Let \(\mathbb H\subset(0,1)\) be a compact parameter set such that \(H_0\in\mathbb H\). In this subsection, we denote by \(\theta\in\mathbb H\) a generic candidate value of the Hurst parameter, in order to avoid confusion with the step size $h$ of the GL operator. For every \(\theta\in\mathbb H\), we construct the GL-filtered unit-scale sample and the GL-filtered multi-scale sample, where the latter is rescaled by \(a^{-\theta}\). We denote the corresponding empirical distribution functions by \(\widetilde F_n\) and \(\widetilde G_{m,\theta}\), respectively. The GL--KS criterion is defined as
\begin{equation}\nonumber
\widetilde D_{n,m}(\theta) = \sup_{x\in\mathbb R} \left| \widetilde F_n(x) - \widetilde G_{m,\theta}(x) \right|.
\end{equation}

The GL--KS estimator of the Hurst parameter is then
\begin{equation}\nonumber
\widehat H_{GLKS} = \operatorname*{\argmin}_{\theta\in\mathds{H}} \widetilde D_{n,m}(\theta). 
\end{equation} 
We now describe the population criterion associated with this estimator. Under the true value \(H_0\), the standardized GL-filtered unit-scale sample has limiting distribution \(\mathcal N(0,1)\). On the other hand, if the multi-scale sample is rescaled by \(a^{-\theta}\), then its limiting variance is proportional to \( a^{2(H_0-\theta)}. \) 

Therefore, the limiting distribution function of the rescaled multi-scale sample is 
\[ x \mapsto \Phi\left(a^{\theta-H_0}x\right). \] 
Hence, the deterministic counterpart of the empirical difference is \[ M(\theta,x) = \Phi(x) - \Phi\left(a^{\theta-H_0}x\right), \] and the corresponding population GL--KS criterion is \[ D(\theta) = \sup_{x\in\mathbb R} \left| M(\theta,x) \right|. \] At the true value \(\theta=H_0\), we have \( M(H_0,x)=0 \) for every \(x\in\mathbb R\), and therefore \( D(H_0)=0. \) Moreover, if \(\theta\neq H_0\), then the two Gaussian marginal distributions are different, so that \( D(\theta)>0. \) Thus \(H_0\) is the unique minimizer of the population criterion. In particular, if \[ \sup_{\theta\in\mathds H} \left| \widetilde D_{n,m}(\theta) - D(\theta) \right| \overset{P}{\longrightarrow} 0, \] then the usual $\argmin$ consistency argument gives \( \widehat H_{GLKS} \overset{P}{\longrightarrow} H_0. \) 

The next proposition gives the local asymptotic distribution of the GL--KS estimator and the corresponding expression for its asymptotic variance.

\begin{proposition}[Local distribution and asymptotic variance of the GL--KS estimator]\label{prop:variance} Let $B^{H_0}_t$ be a fractional Brownian motion with $H_0\in(1/2,1)$. Let the GL filter order $\alpha>0$ satisfy $ H_0-\alpha<\frac12.$ Assume that the conditions of Theorem \ref{thm:GL_KS_limit} and Lemma \ref{lem:burn-in} hold. Assume also that the empirical process is locally asymptotically equicontinuous with respect to the parameter $\theta$ in a neighbourhood of $H_0$. 

Let $r_{n,m} = \sqrt{\frac{nm}{n+m}}.$ If $H_0\in\operatorname{int}(\mathds H)$ and the limiting $\argmin$ below is almost surely unique, then 
\[ r_{n,m} \left( \widehat H_{GLKS} - H_0 \right) \Rightarrow T_{KS}, \]
where $T_{KS} = \operatorname*{\argmin}\limits_{t\in\mathbb R} \left\| \widetilde U - t\ell_a \right\|_\infty, $ with $\ell_a(x) = (\log a)x\phi(x),$ and where $\widetilde U$ is the centered Gaussian process appearing in Theorem \ref{thm:GL_KS_limit}. Consequently, if the sequence \(r_{n,m}^2(\widehat H_{GLKS}-H_0)^2\) is uniformly integrable, then \[ r_{n,m}^2 \operatorname{Var} \left( \widehat H_{GLKS} \right) \longrightarrow \sigma_{KS}^2, \] where \( \sigma_{KS}^2 = \operatorname{Var} \left( T_{KS} \right). \) Equivalently, \[ \operatorname{Var} \left( \widehat H_{GLKS} \right) = \frac{n+m}{nm} \sigma_{KS}^2 + o\left( \frac{n+m}{nm} \right). \] In the balanced case $m\simeq n$, this becomes $\operatorname{Var} \left( \widehat H_{GLKS} \right) = \frac{2}{n} \sigma_{KS}^2 + o\left( \frac{1}{n} \right).$ \end{proposition} \begin{proof} We study the local behavior of the criterion in a shrinking neighbourhood of \(H_0\). The deterministic component of the difference between the two limiting distribution functions is 
\[ M(\theta,x) = \Phi(x) - \Phi\left(a^{\theta-H_0}x\right). \]
Differentiating with respect to $\theta$, we obtain 
\[ \frac{\partial}{\partial\theta} M(\theta,x) = - (\log a) a^{\theta-H_0} x \phi\left(a^{\theta-H_0}x\right). \] 
Therefore, at $\theta=H_0$, 
\[ \dot M_{H_0}(x) = \left. \frac{\partial}{\partial\theta} M(\theta,x) \right|_{\theta=H_0} = - (\log a)x\phi(x). \] 
Set $ \ell_a(x) = (\log a)x\phi(x).$ Then $\dot M_{H_0}(x) = -\ell_a(x).$ Hence, uniformly in \(x\in\mathbb R\), \[ M(\theta,x) = - (\theta-H_0)\ell_a(x) + o\left( |\theta-H_0| \right),\quad \text{as}\quad \theta\to H_0.\] 

Now take a local perturbation of the form $\theta = H_0 + \frac{t}{r_{n,m}}.$ Then \[ r_{n,m} M\left( H_0+\frac{t}{r_{n,m}}, x \right) \longrightarrow - t\ell_a(x). \] 
By Theorem \ref{thm:GL_KS_limit}, at the true parameter value, \[ r_{n,m} \left[ \widetilde F_n(x) - \widetilde G_{m,H_0}(x) \right] \Rightarrow \widetilde U(x) \qquad \text{in } \ell^\infty(\mathbb R).\] Combining this convergence with the previous local expansion and the assumed local stochastic equicontinuity gives, for every compact set \(K\subset\mathbb R\), \[ r_{n,m} \widetilde D_{n,m} \left( H_0+\frac{t}{r_{n,m}} \right) \Rightarrow \left\| \widetilde U - t\ell_a \right\|_\infty \qquad \text{in } \ell^\infty(K).\] Since $H_0$ is the unique minimizer of the population criterion and the limiting minimizer is assumed to be almost surely unique, the $\argmin$ continuous mapping theorem yields \[r_{n,m} \left( \widehat H_{GLKS} - H_0 \right) \Rightarrow \operatorname*{\argmin}_{t\in\mathbb R} \left\| \widetilde U - t\ell_a \right\|_\infty.\] This proves the claimed weak convergence, with $T_{KS} = \operatorname*{\argmin}_{t\in\mathbb R} \left\| \widetilde U - t\ell_a \right\|_\infty.$ The variance statement follows from the convergence of second moments. If \(r_{n,m}^2(\widehat H_{GLKS}-H_0)^2\) is uniformly integrable, then $r_{n,m}^2 \operatorname{Var} \left( \widehat H_{GLKS} \right) \longrightarrow \operatorname{Var} \left( T_{KS} \right).$ Since $r_{n,m}^{-2} = \frac{n+m}{nm},$ we obtain \[ \operatorname{Var} \left( \widehat H_{GLKS} \right) = \frac{n+m}{nm} \operatorname{Var} \left( T_{KS} \right) + o\left( \frac{n+m}{nm} \right). \] This concludes the proof. \end{proof}

\begin{remark}[The short-memory case as the limiting case ($\alpha=0$)]\label{rmk:GL_for_SRD}
    Theorem~\ref{thm:GL_KS_limit} has been stated for the long-memory case $H_0>1/2$, where the Gr\"{u}nwald--Letnikov derivative is needed in order to bring the empirical process back to the short-memory domain. However, the same local argument also covers the case $H_0\le 1/2$ by setting the fractional order equal to zero.

Indeed, when $H_0<1/2$, the fGn autocovariance sequence is absolutely summable, while for $H_0=1/2$ the increments are independent. Therefore, the obstruction created by long-range dependence is absent from the outset. In this case one may set
\[\alpha=0, \Delta^{GL,0}_h=(1-L_h)^0=I,\] so that the GL-filtered statistic reduces to the unfiltered KS-type statistic. No pre-sample fractional-filter truncation is introduced, and the burn-in correction can be omitted. Consequently, the local expansion of the population criterion remains unchanged:
\[
M(\theta,x)
=
\Phi(x)-\Phi\left(a^{\theta-H_0}x\right),
\qquad
\dot M_{H_0}(x)
=
-(\log a)x\phi(x).
\]
The same argmin argument therefore gives
\[
r_{n,m}\left(\widehat H_{\mathrm{KS}}-H_0\right)
\Rightarrow
\argmin_{t\in\mathbb R}
\left\|
U-t\ell_a
\right\|_\infty,
\qquad
\ell_a(x)=(\log a)x\phi(x),
\]
where $U$ denotes the centered Gaussian limit of the unfiltered short-memory empirical process specified in Proposition \ref{prop:SRD}. Thus, the variance formula in Proposition~\ref{prop:variance} remains valid in the short-memory and independent regimes after replacing the filtered limiting process by its unfiltered counterpart and taking $\alpha=0$.

\end{remark}

\begin{remark}\label{rmk:role_scaling}
The scaling factor \(a\) plays a direct role in the local identification of the Hurst parameter. Around the true value \(H_0\), the deterministic component of the criterion satisfies
\[
M(\theta,x)
=
-(\theta-H_0)(\log a)x\phi(x)
+
o(|\theta-H_0|).
\]
Hence the local drift of the criterion is proportional to \(\log a\). Larger values of \(a\) increase the separation between the unit-scale distribution and the rescaled multi-scale distribution when \(\theta\neq H_0\). In this sense, increasing \(a\) steepens the local objective function and tends to improve the precision of the Hurst estimate.

This effect is also visible in the limiting argmin representation
\[
T_{\mathrm{KS}}(a)
=
\argmin_{t\in\mathbb R}
\left\|
\widetilde U-t\ell_a
\right\|_\infty,
\qquad
\ell_a(x)=(\log a)x\phi(x).
\]
If the covariance structure of the limiting process were kept fixed as $a$ varies, the deterministic slope would imply an approximate variance reduction of order $(\log a)^{-2}$. Therefore, from a purely local-identification perspective, larger scales provide a stronger signal for distinguishing nearby Hurst parameters.

In finite samples, however, this gain is counterbalanced by the structure of the crossed sample. As $a$ increases, each branch $r=0,\ldots,a-1$ contains fewer observations, and the burn-in deletion is applied branch by branch. Thus the effective multi-scale sample size $m_{\mathrm{eff}}
=
\sum_{r=0}^{a-1}(m_r-M_{0,r})$ may decrease relative to the nominal pooled size, especially for large $a$ and moderate $N$. Very large scales may therefore improve identification through the factor $\log a$, but at the cost of shorter branches, stronger finite-sample effects, and a smaller effective sample after burn-in correction.

The choice of $a$ should consequently be interpreted as a bias--variance and identification--sample-size trade-off. Moderate values of $a$ are expected to improve power and reduce estimator dispersion, whereas excessively large values can deteriorate finite-sample stability because the crossed branches become too short.
    
\end{remark}

\begin{remark}[Regime-adaptive implementation]\label{rmk:regime_adaptive}
The GL filter is not used to redefine the Hurst exponent. It is used only to select a fast-converging limiting distribution when the estimated regime is persistent. In the empirical implementation in Section \ref{sec:Application}, we first estimate $H$ from
the distributional self-similarity criterion. If $\widehat H\le 1/2$, we set $\alpha=0$, use the short-memory limit, and no burn-in correction is applied. If $\widehat H>1/2$, we choose $\alpha$ such that $\widehat H-\alpha<1/2$, apply the GL filter, and compute the GL-KS statistic with the branch-wise effective sample sizes. This plug-in regime selection is asymptotically valid away from the boundary $H=1/2$, because $\widehat H\to H_0$ in probability.
\end{remark}

The theoretical results above are asymptotic. The next section verifies whether the GL-KS statistic retains its predicted stabilization in finite samples, whether the burn-in deletion is sufficient to remove pre-sample truncation effects, and whether the estimator remains reliable across both LRD and SRD regimes.

\section{Computational Analysis}
\label{sec:computational_analysis}

In this section, we present a comprehensive Monte Carlo simulation study designed with a two-fold objective: first, to empirically validate the asymptotic invariance and phase-transition bypass established in Theorem \ref{thm:GL_KS_limit}; and second, to evaluate the finite-sample performance of the Lemma \ref{lem:burn-in} boundary conditions in mitigating pre-history truncation bias. The computational analysis systematically contrasts the standard KS test against our proposed fractional Gr\"{u}nwald-Letnikov filtered alternative GL-KS.

\subsection{General Monte Carlo Design}
\label{subsec:setup}
To ensure statistical exactness and prevent the structural distortions often induced by Cholesky factorizations or truncated spectral methods, the data generating process employs the circulant embedding technique pioneered by Wood and Chan in \citep{WoodChan1994}. This approach enables the exact synthesis of fBm trajectories by decomposing the associated embedded circulant covariance matrix. It preserves the underlying LRD structure intact while maintaining an optimal computational complexity of $\mathcal{O}(N \log N)$ via the Fast Fourier Transform. Any fractional Brownian motion was generated using \textit{fbmwoodchan()}  MATLAB function, available in the FracLab Toolbox 2.02 (INRIA package).

The experimental parameter space for the Monte Carlo framework is calibrated as follows:
\begin{itemize}
    \item \textit{Hurst exponent}. In the LRD regime we tested the true Hurst exponent $H_{true}=\{0.51,0.7,0.9\}$, while in the SRD regime we tested $H_{true}=\{0.1,0.2,0.3,0.4,0.5\}$. The candidate Hurst values used in the minimization procedure are evaluated on the grid $\mathds H\in[0.001,0.999]$ with step $\Delta \mathds H = 0.001$.
    \item \textit{Fractional GL order} $\alpha$. In the LRD regime, the differencing parameter $\alpha$ is chosen to strictly enforce the short memory regime. Specifically, $\alpha$ is selected such that $H - \alpha = 0.25$ to enable a consistent comparison across all simulations. In the SRD regime, we use the unfiltered short-memory implementation of Remark~\ref{rmk:GL_for_SRD}.
    \item \textit{Burn-in exponent} $\gamma$. For the GL-KS statistic, $\gamma$ is selected above the theoretical threshold of Lemma~\ref{lem:burn-in} to control the pre-sample truncation error. Unless otherwise stated, its value is fixed within each experiment; in the burn-in sensitivity analysis, $\gamma$ is varied over a grid. No burn-in correction is applied when $\alpha=0$.
    \item \textit{Sample size} $N$. To investigate the rate of asymptotic convergence, fBm trajectories are simulated across five progressive lengths: $N \in \{100, 250, 500, 1000, 5000\}$.
    \item \textit{Scaling factor} $a$. In alignment with the underlying cross-scaled fGn framework, the scaling factor is varied over $a \in \{2,3,4,5,10,20\}$.
    \item \textit{Monte Carlo replications} $\mathcal M$. For each parameter configuration, we perform $\mathcal M=10{,}000$ independent Monte Carlo replications. This number of replications provides stable estimates of the empirical distributions, rejection frequencies, and critical quantiles used in the finite-sample analysis.
\end{itemize}

\subsection{Convergence Acceleration and Asymptotic Invariance}
\label{subsec:convergence}

The first numerical experiment is dedicated to demonstrating the collapse of the LRD-driven phase transition discussed in Theorem \ref{thm:GL_KS_limit}. As documented in the classical literature (e.g., Dehling and Taqqu, 1989), applying non-parametric goodness-of-fit tests directly to long-memory processes compromises the standard convergence rate—slowing it down to $n^{1-H}$—and causes the limiting distribution to become degenerate and structurally dependent on the unknown parameter $H$.

By introducing the fractional GL filter and scaling the resultant $\tilde{D}_{n,m}$ statistic by the standard $\sqrt{n}$ rate, asymptotic invariance is successfully rehabilitated. We monitor this phenomenon by tracking the evolution of the ECDFs as $N$ increases.

\begin{figure}[!ht]
    \centering
    \includegraphics[width=0.85\linewidth]{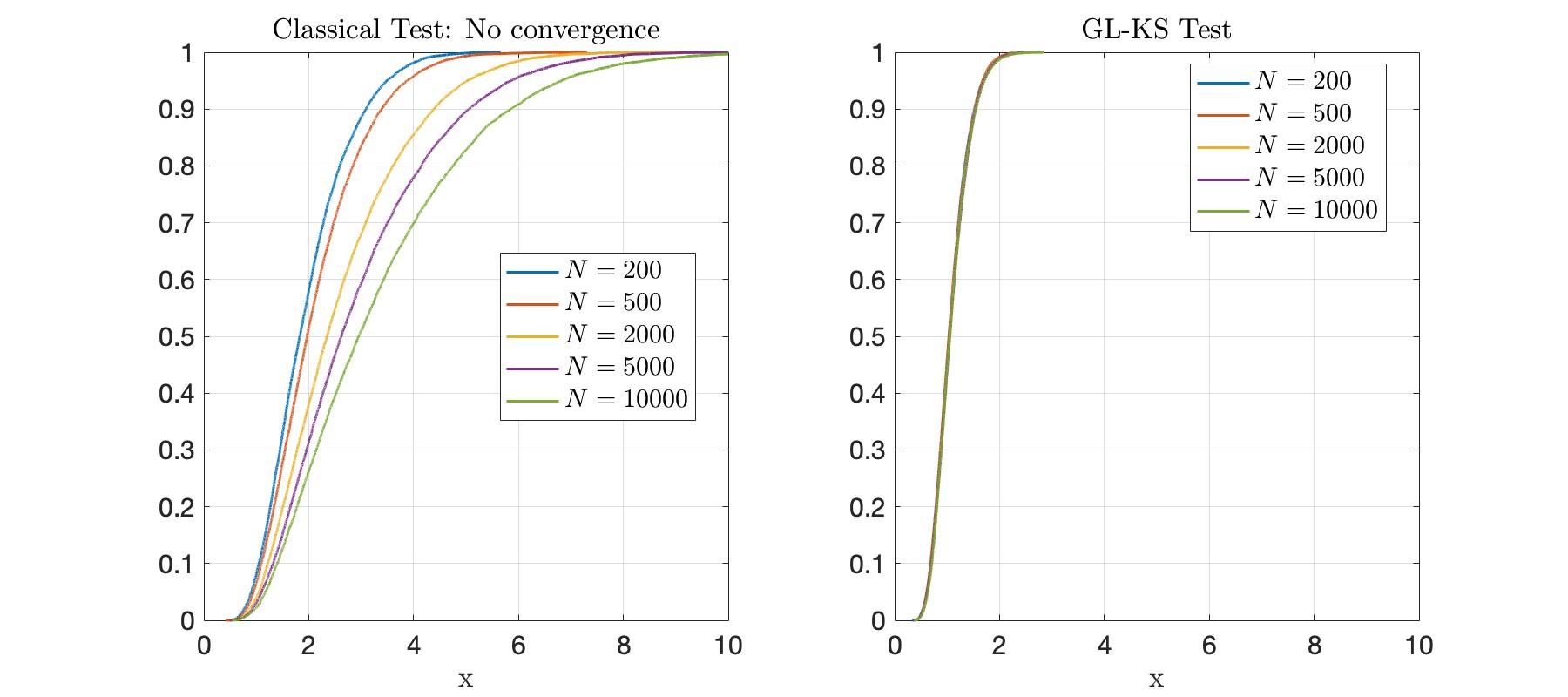}
    \caption{\scriptsize Empirical null distributions of the normalized KS-type statistic for increasing sample sizes. The left panel reports the classical unfiltered statistic under long-range dependence, showing the lack of stabilization under the standard $\sqrt n$ normalization. The right panel reports the GL-KS statistic after fractional filtering, whose empirical distributions collapse onto a common limiting curve, consistently with the short-memory asymptotic regime restored by the GL derivative.}
    \label{fig:normalization}
\end{figure}
Figure~\ref{fig:normalization} illustrates the stabilization effect predicted by Theorem~\ref{thm:GL_KS_limit}. In the left panel, the unfiltered classical statistic exhibits ECDF curves that shift systematically to the right as $N$ expands, verifying that the $\sqrt{n}$ normalization is mathematically misspecified under LRD. Conversely, the right panel illustrates that the GL-KS test statistics are much more stable across sample sizes, even at a moderate sample size of $N=200$. This provides visual evidence that the fractional filter successfully neutralizes the singularity at the origin of the spectral density, effectively relocating the process into the domain of attraction of short-memory functionals.

\subsection{Finite-Sample Properties: Size and Power}
\label{subsec:size_power}

We now investigate the finite-sample inferential performance of the proposed statistic. The analysis focuses on two complementary aspects: empirical size under the null hypothesis and power against misspecified Hurst exponents.

For each Monte Carlo replication, a fractional Brownian motion with parameter $H_{\mathrm{true}}$ is generated, and the corresponding unit-scale and crossed multi-scale increments are constructed as in Section~\ref{subsec:setup}. The null hypothesis is $\mathcal H_0:H_{\mathrm{test}}=H_{\mathrm{true}},$ where $H_{\mathrm{test}}$ is the value used to rescale the multi-scale sample. Under the null, the rescaled unit-scale and multi-scale samples should have the same distribution. Rejections are therefore driven by the two-sample KS-type distance between the corresponding empirical distribution functions.

In the persistent case $H_{\mathrm{true}}>1/2$, two statistics are compared. The first one is the unfiltered statistic, which is affected by the long-memory phase transition described in Proposition~\ref{prop:LRD}. The second one is the GL--KS statistic. For the filtered statistic, the GL order is chosen so that the effective memory parameter satisfies $H_{\mathrm{true}}-\alpha=0.25,$ and the burn-in exponent is selected above the lower bound of Lemma~\ref{lem:burn-in}. This ensures that the filtered process lies in the short-memory domain while preserving the self-similarity parameter identified by the KS criterion.

\begin{table}[!ht]
\centering
\caption{Empirical size in the long-memory regime.}
\label{tab:empirical_size}
\begin{scriptsize}
\begin{tabular}{c|c|cccccc}
\toprule
\hline
& & \multicolumn{3}{c}{\textbf{Unfiltered KS Test}} & \multicolumn{3}{c}{\textbf{GL-KS Test}} \\
\cline{3-8}
\multirow{2}{*}{Sign. Level}& & $H=0.51$ & $H=0.70$ & $H=0.90$ & $H=0.51$ & $H=0.70$ & $H=0.90$ \\
&\diagbox{$N$}{$\alpha$} & - & - & - & $0.26$ & $0.45$ & $0.65$\\
\hline
\multirow{5}{*}{$1\%$}
&$100$  & $4.44\%$ & $7.50\%$ & $47.12\%$ & $0.90\%$  & $1.08\%$ & $1.06\%$\\
&$250$  & $5.35\%$ & $6.07\%$ & $30.30\%$ & $1.16\%$ & $1.21\%$ & $0.98\%$\\
&$500$  & $5.84\%$ & $4.74\%$ & $19.49\%$ & $1.17\%$ & $1.16\%$ & $1.14\%$\\
&$1000$  & $6.87\%$ & $3.71\%$ & $12.31\%$ & $0.96\%$ & $1.08\%$ & $1.04\%$\\
&$5000$  & $6.19\%$ & $2.40\%$ & $5.14\%$ & $0.90\%$ & $1.06\%$ & $1.06\%$ \\
\hline
\multirow{5}{*}{$5\%$}
&$100$  & $17.48\%$ & $24.26\%$ & $76.95\%$ & $4.82\%$ & $5.11\%$ & $5.09\%$\\
&$250$ & $19.75\%$ & $19.30\%$ & $58.91\%$ & $4.80\%$ & $5.04\%$ & $5.01\%$\\
&$500$ & $20.54\%$ & $16.95\%$ & $43.78\%$ & $5.36\%$ & $5.04\%$ & $4.83\%$\\
&$1000$ & $22.10\%$ & $14.36\%$ & $31.33\%$  & $4.70\%$ & $5.12\%$ & $5.14 \%$\\
&$5000$ & $21.39\%$ & $9.45\%$ & $16.27\%$ & $ 4.95\%$ & $ 4.94\%$ & $5.22\%$\\
\hline
\multirow{5}{*}{$10\%$}
&$100$  & $30.39\%$ & $39.82\%$ & $88.51\%$ & $9.85\%$ & $9.96\%$ & $10.17\%$\\
&$250$  & $34.29\%$ & $33.30\%$ & $75.87\%$ & $9.99\%$ & $10.23\%$ & $9.77\%$\\
&$500$  & $35.26\%$ & $29.64\%$ & $61.28\%$ & $10.43\%$ & $10.15\%$ & $10.28\%$\\
&$1000$  & $36.57\%$ & $25.17\%$ & $46.74\%$ & $10.15\%$ & $10.47\%$ & $10.22\%$\\
&$5000$  & $35.88\%$ & $17.51\%$ & $27.79\%$ & $10.10\%$ & $10.12\%$ & $10.58\%$ \\
\bottomrule
\end{tabular}
\end{scriptsize}
\vspace{0.15cm} \begin{minipage}{0.94\textwidth} \footnotesize \emph{Notes.} Rejection frequencies are reported at the \(1\%\), \(5\%\), and \(10\%\) nominal levels under \(\mathcal H_0:H_{\mathrm{test}}=H_{\mathrm{true}}\), with \(a=20\). The unfiltered statistic is compared with the GL--KS statistic. In the GL--KS columns, \(\alpha\) is chosen so that \(H_{\mathrm{true}}-\alpha=0.25\), and the burn-in correction is applied branch-wise. The analysis was performed with $\mathcal M = 10,000$ Monte Carlo simulations.\end{minipage}
\end{table}

Critical values are obtained by Monte Carlo calibration under the corresponding null configuration. This is important because the limiting law is not the classical distribution-free two-sample KS law: the two samples are extracted from the same trajectory and the limiting covariance depends on the auto-covariance and cross-covariance structure of the filtered empirical process. The purpose of the simulation is therefore to compare the finite-sample rejection frequencies of the unfiltered and filtered statistics at the same nominal levels. Table~\ref{tab:empirical_size} reports empirical rejection frequencies at the \(1\%\), \(5\%\), and \(10\%\) nominal levels for \(H_{\mathrm{true}}\in\{0.51,0.70,0.90\}\), with scaling factor \(a=20\). The unfiltered statistic displays substantial size distortions in the long-memory regime. The distortion is especially severe for large values of \(H\), where the autocorrelation function decays very slowly and the standard \(\sqrt n\)-type empirical-process approximation is no longer appropriate. For example, when \(H=0.90\), the unfiltered rejection rates remain far above their nominal levels even for large sample sizes. By contrast, the GL--KS statistic remains well calibrated across all three persistent regimes. Its empirical rejection frequencies stay close to the nominal \(1\%\), \(5\%\), and \(10\%\) levels, with only minor finite-sample fluctuations. This confirms the main theoretical implication of Theorem~3.7: after fractional filtering, the statistic returns to a short-memory empirical-process regime and the standard \(\sqrt n\) normalization becomes appropriate again.

\begin{figure}[!ht]
    \centering
    \includegraphics[width=0.49\linewidth]{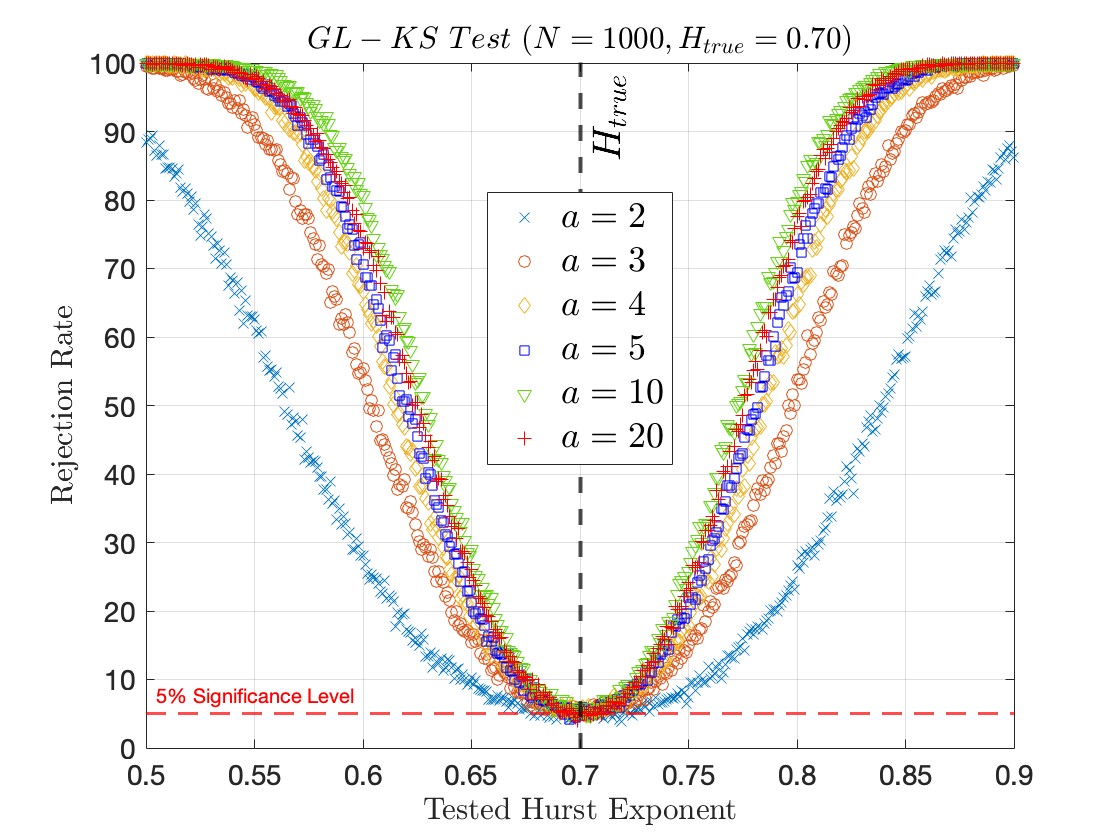}
    \includegraphics[width=0.49\linewidth]{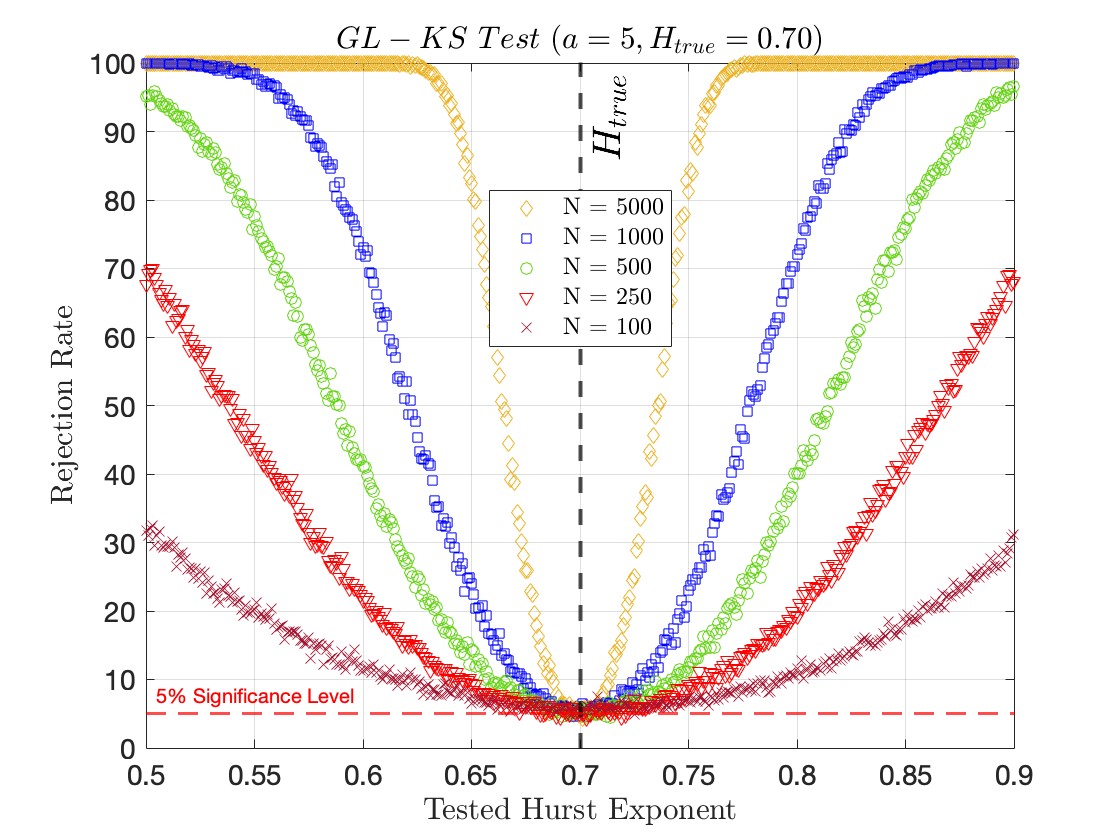}
    \caption{\scriptsize Empirical power of the GL-KS test against misspecified Hurst exponents. The true value $H_{\mathrm{true}}$ is marked by the vertical dashed line, while the horizontal dashed line denotes the 5\% significance level. The left panel shows the effect of the scaling factor $a$ for fixed sample size, whereas the right panel shows the sharpening of the rejection curve as the sample size $N$ increases.}
    \label{fig:PowerTest}
\end{figure}

The second experiment evaluates power. We fix the true value at $H_{\mathrm{true}}=0.70$ and vary the null value $H_{\mathrm{test}}$ over a grid around the true parameter. For each tested value, the multi-scale sample is rescaled by $a^{-H_{\mathrm{test}}}$. Hence, when $H_{\mathrm{test}}\neq H_{\mathrm{true}}$, the unit-scale and multi-scale samples no longer have the same distribution, and the test should reject. In the persistent part of the grid, the GL order is selected under the tested null so that $H_{\mathrm{test}}-\alpha=0.25.$ At the boundary and in the short-memory part of the grid, $\alpha=0$ is used, consistently with Remark~\ref{rmk:GL_for_SRD}. Figure~\ref{fig:PowerTest} reports the resulting empirical power curves. The minimum of the rejection curve occurs at $H_{\mathrm{test}}=H_{\mathrm{true}}$, where the rejection probability is close to the nominal significance level. As $H_{\mathrm{test}}$ moves away from $H_{\mathrm{true}}$, the rejection probability increases, producing the expected V-shaped power profile. The curve becomes steeper as the sample size increases, showing that the test becomes increasingly sensitive to local departures from the null. This behavior is consistent with the local identification argument in Proposition~\ref{prop:variance}: the deterministic drift of the criterion separates incorrect Hurst values from the true one, while the stochastic component shrinks at the standard empirical-process rate after GL filtering. The power experiment also highlights the role of the scaling factor \(a\). Larger values of \(a\) increase the deterministic separation between the unit-scale distribution and the incorrectly rescaled multi-scale distribution, because the local drift is proportional to \(\log a\). At the same time, very large scales reduce branch lengths and may increase finite-sample variability. The numerical evidence therefore supports the bias--variance interpretation discussed in Remark~\ref{rmk:role_scaling}: moderate-to-large scales improve identification, but excessively large scales can reduce effective sample information.

\subsection{Sensitivity to Burn-in Initialization}
\label{subsec:burn_in_sensitivity}

The final component of the computational analysis examines the finite-sample role of the burn-in exponent. Lemma~\ref{lem:burn-in} requires
$\gamma>\gamma_{\min}:=\frac{1}{2(1+\alpha-H)}$ in order for the pre-sample truncation error to be negligible at the empirical-process scale. In this sensitivity experiment we set $H=0.70$ and $\alpha=0.50$, so that $H-\alpha=0.20$ and $\gamma_{\min}=5/8$. We therefore vary $\gamma$ over a grid containing values below and above this threshold and record the empirical rejection frequency under the null hypothesis.

Figure~\ref{fig:gamma} shows the expected finite-sample trade-off. For $\gamma<\gamma_{\min}$, the burn-in deletion is not large enough to absorb the initialization error generated by the finite GL filter. The empirical size is therefore distorted. Once $\gamma$ crosses the theoretical lower bound, the rejection frequency stabilizes around the significance level, confirming the practical relevance of Lemma~\ref{lem:burn-in}. For very large values of $\gamma$, however, the effective sample size is unnecessarily reduced. This last deterioration is not a contradiction of the lemma: the lemma only gives a lower bound ensuring asymptotic negligibility of the truncation error, while finite samples also require preserving enough observations after the branch-wise deletion.

\begin{figure}[!ht]
    \centering
    \includegraphics[width=0.5\linewidth]{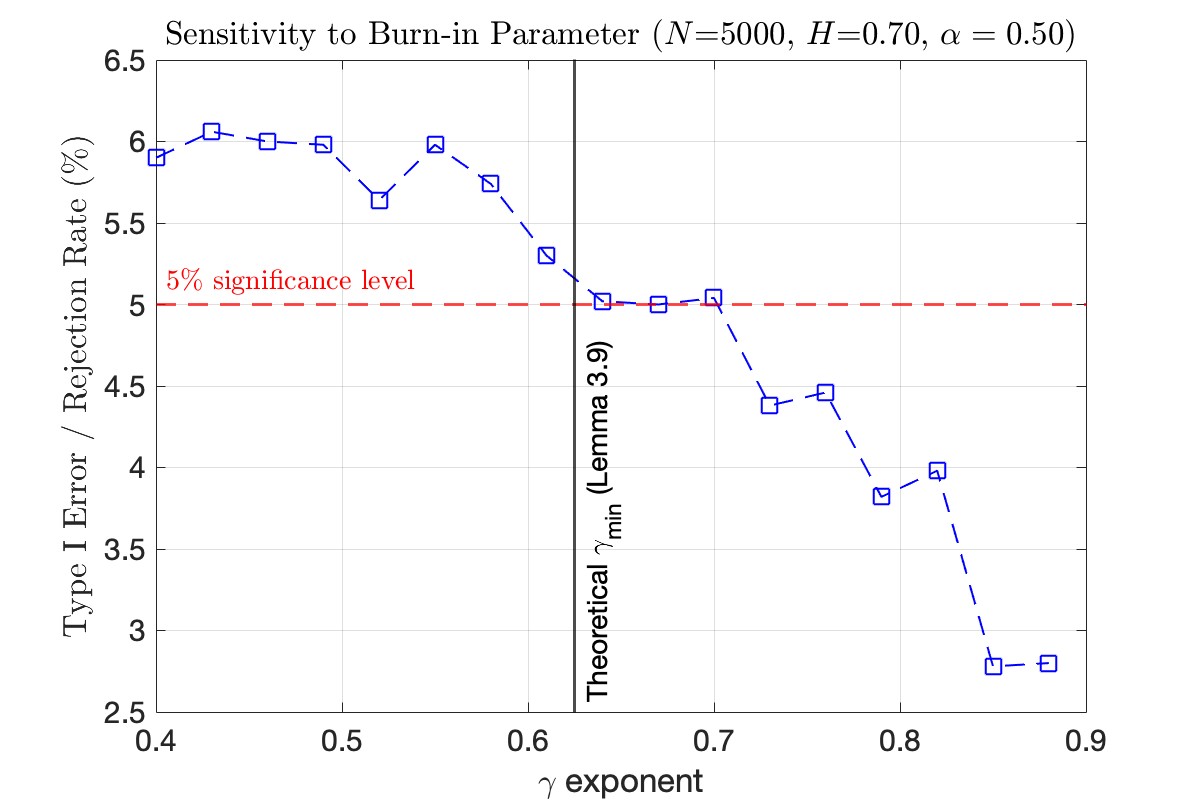}
    \caption{\scriptsize Sensitivity of the empirical size to the burn-in exponent $\gamma$. The vertical line marks the theoretical lower bound $\gamma_{\min}=1/[2(1+\alpha-H)]$, while the horizontal dashed line denotes the $5\%$ level. For $\gamma<\gamma_{\min}$, the pre-sample truncation error is not sufficiently absorbed by the burn-in deletion and the test over-rejects. Once $\gamma$ exceeds the theoretical threshold, the empirical size stabilizes around the nominal level; for very large $\gamma$, the effective sample size becomes unnecessarily depleted.}
    \label{fig:gamma}
\end{figure}

\subsection{Short-Memory Benchmark and Finite-Sample Accuracy of KS/GL-KS test}

We finally complete the Monte Carlo analysis by considering the short-memory and independent regimes \(H\leq1/2\), and by evaluating the finite-sample accuracy of the Hurst estimator. This serves two purposes: it verifies that the size distortions observed for \(H>1/2\) are genuinely caused by long-range dependence, and it validates the regime-adaptive rule of Remark~\ref{rmk:regime_adaptive}.

Accordingly, in the short-memory experiment we simulate fBm trajectories with \\$H_{\mathrm{true}}\in\{0.1,0.2,0.3,0.4,0.5\}.$ For each value of $H_{\mathrm{true}}$, the null hypothesis is $\mathcal H_0:H_{\mathrm{test}}=H_{\mathrm{true}}.$ The empirical statistic is the unfiltered short-memory statistic $D_{n,m}^{\star}=\sqrt{\frac{nm}{n+m}}\,D_{n,m},$ corresponding to the case $\alpha=0$ and $\Delta_h^{GL,0}=I$. Since no fractional filter is applied,
there is no pre-sample truncation error and therefore no burn-in deletion. The effective sample sizes coincide with the nominal crossed-sample sizes.

\begin{table}[!ht]
\centering
\caption{Empirical size in the short-memory and independent regimes.}
\label{tab:SRD}
\begin{scriptsize}
\begin{tabular}{c|c|ccccc}
\toprule
\hline
\cline{3-7}
Sign. Level& $N$ & $H=0.1$ & $H=0.2$ & $H=0.3$ & $H=0.4$ & $H=0.5$ \\
\hline
\multirow{5}{*}{$1\%$}
&$100$  & $1.13\%$ & $0.88\%$ & $1.07\%$ & $0.91\%$  & $1.05\%$ \\
&$250$  & $0.87\%$ & $1.03\%$ & $1.05\%$ & $1.00\%$ & $1.09\%$ \\
&$500$  & $0.99\%$ & $1.02\%$ & $1.00\%$ & $0.97\%$ & $1.06\%$ \\
&$1000$  & $1.08\%$ & $1.10\%$ & $1.10\%$ & $0.97\%$ & $1.06\%$\\
&$5000$  & $1.18\%$ & $1.07\%$ & $1.01\%$ & $0.92\%$ & $1.12\%$  \\
\hline
\multirow{5}{*}{$5\%$}
&$100$  & $5.21\%$ & $4.94\%$ & $5.03\%$ & $5.03\%$ & $4.94\%$ \\
&$250$ & $4.91\%$ & $4.90\%$ & $4.96\%$ & $4.91\%$ & $4.77\%$ \\
&$500$ & $5.34\%$ & $5.14\%$ & $4.87\%$ & $5.04\%$ & $5.21\%$ \\
&$1000$ & $5.14\%$ & $4.98\%$ & $4.91\%$  & $5.17\%$ & $5.10\%$ \\
&$5000$ & $4.99\%$ & $4.61\%$ & $5.14\%$ & $ 4.84\%$ & $ 4.90\%$\\
\hline
\multirow{5}{*}{$10\%$}
&$100$  & $10.46\%$ & $9.91\%$ & $9.78\%$ & $9.73\%$ & $9.89\%$ \\
&$250$  & $9.96\%$ & $10.56\%$ & $9.77\%$ & $10.13\%$ & $10.29\%$ \\
&$500$  & $10.28\%$ & $10.17\%$ & $9.82\%$ & $9.84\%$ & $9.84\%$ \\
&$1000$  & $9.98\%$ & $10.06\%$ & $10.08\%$ & $9.95\%$ & $10.19\%$ \\
&$5000$  & $10.22\%$ & $9.89\%$ & $10.18\%$ & $9.90\%$ & $9.81\%$ \\
\bottomrule
\end{tabular}
\end{scriptsize}
\vspace{0.15cm} \begin{minipage}{0.94\textwidth} \footnotesize \emph{Notes.} Rejection frequencies are reported at the \(1\%\), \(5\%\), and \(10\%\) nominal levels under \(\mathcal H_0:H_{\mathrm{test}}=H_{\mathrm{true}}\). Since \(H_{\mathrm{true}}\leq1/2\), the unfiltered short-memory statistic is used and no burn-in deletion is required. The analysis was performed with $\mathcal M = 10,000$ Monte Carlo simulations.\end{minipage}
\end{table}

Table~\ref{tab:SRD} shows that the unfiltered KS-type statistic is already well calibrated throughout the short-memory region. At all nominal significance levels, the empirical rejection frequencies remain close to their theoretical targets. This holds both in the anti-persistent cases $H<1/2$ and in the independent Brownian benchmark $H=1/2$.

\begin{table}[!ht]
\centering

\caption{Finite-sample accuracy of the Hurst estimator.}
\label{tab:hhat-accuracy}
\resizebox{\textwidth}{!}{%
\tiny
\begin{tabular}{rrrrrrrrrrrr}
\toprule
$H_0$ & $N$ & $a$ & $\alpha$ & $\gamma$ & $n_{\rm eff}$ & $m_{\rm eff}$ & Mean & Bias & Std & RMSE & MAE \\
\midrule

 0.10 & 1000 & 20 & 0.000 & -- & 999 & 980 & 0.1000 & -0.0000 & 0.0196 & 0.0196 & 0.0154 \\
 0.20 & 1000 & 20 & 0.000 & -- & 999 & 980 & 0.2006 & 0.0006 & 0.0248 & 0.0248 & 0.0193 \\
 0.30 & 1000 & 20 & 0.000 & -- & 999 & 980 & 0.3006 & 0.0006 & 0.0296 & 0.0296 & 0.0236 \\
 0.40 & 1000 & 20 & 0.000 & -- & 999 & 980 & 0.4018 & 0.0018 & 0.0342 & 0.0342 & 0.0271 \\
 0.50 & 1000 & 20 & 0.000 & -- & 999 & 980 & 0.5018 & 0.0018 & 0.0389 & 0.0389 & 0.0304 \\
 0.51 & 1000 & 20 & 0.260 & 0.697 & 877 & 681 & 0.5097 & -0.0003 & 0.0421 & 0.0421 & 0.0339 \\
 0.60 & 1000 & 20 & 0.350 & 0.697 & 877 & 681 & 0.6006 & 0.0006 & 0.0451 & 0.0451 & 0.0356 \\
 0.70 & 1000 & 20 & 0.450 & 0.697 & 877 & 681 & 0.6990 & -0.0010 & 0.0462 & 0.0461 & 0.0370 \\
 0.80 & 1000 & 20 & 0.550 & 0.697 & 877 & 681 & 0.8044 & 0.0044 & 0.0492 & 0.0494 & 0.0395 \\
 0.90 & 1000 & 20 & 0.650 & 0.697 & 877 & 681 & 0.9019 & 0.0019 & 0.0478 & 0.0478 & 0.0386 \\
 0.10 & 5000 & 20 & 0.000 & -- & 4999 & 4980 & 0.1000 & -0.0000 & 0.0087 & 0.0087 & 0.0069 \\
 0.20 & 5000 & 20 & 0.000 & -- & 4999 & 4980 & 0.1996 & -0.0004 & 0.0108 & 0.0108 & 0.0086 \\
 0.30 & 5000 & 20 & 0.000 & -- & 4999 & 4980 & 0.3001 & 0.0001 & 0.0127 & 0.0127 & 0.0100 \\
 0.40 & 5000 & 20 & 0.000 & -- & 4999 & 4980 & 0.4003 & 0.0003 & 0.0150 & 0.0150 & 0.0120 \\
 0.50 & 5000 & 20 & 0.000 & -- & 4999 & 4980 & 0.5008 & 0.0008 & 0.0170 & 0.0170 & 0.0135 \\
 0.51 & 5000 & 20 & 0.260 & 0.697 & 4623 & 4061 & 0.5100 & 0.0000 & 0.0175 & 0.0175 & 0.0138 \\
 0.60 & 5000 & 20 & 0.350 & 0.697 & 4623 & 4061 & 0.6004 & 0.0004 & 0.0174 & 0.0174 & 0.0139 \\
 0.70 & 5000 & 20 & 0.450 & 0.697 & 4623 & 4061 & 0.6993 & -0.0007 & 0.0192 & 0.0192 & 0.0152 \\
 0.80 & 5000 & 20 & 0.550 & 0.697 & 4623 & 4061 & 0.7990 & -0.0010 & 0.0194 & 0.0194 & 0.0155 \\
 0.90 & 5000 & 20 & 0.650 & 0.697 & 4623 & 4061 & 0.8999 & -0.0001 & 0.0200 & 0.0200 & 0.0158 \\
\bottomrule
\end{tabular}%
}
\vspace{0.15cm} \begin{minipage}{0.94\textwidth} \footnotesize \emph{Notes.} The table reports Monte Carlo summaries of \(\widehat H\): mean, bias, standard deviation, RMSE, and MAE. For \(H_0\leq1/2\), the estimator is computed with the unfiltered KS criterion. For \(H_0>1/2\), it is computed with the GL--KS criterion, with \(\alpha\) chosen so that \(H_0-\alpha=0.25\), and with branch-wise burn-in correction. The analysis uses \(\mathcal M=1000\) Monte Carlo simulations. \end{minipage}
\end{table}

The result is important for the interpretation of the full procedure. In the short-memory regime, the autocorrelation function is summable, so the empirical process converges at the standard $\sqrt n$ rate. The phase transition observed for $H>1/2$ is therefore not a finite-sample artifact of the KS criterion itself, but a consequence of the non-summable dependence structure of the underlying increments. Once the process is already short-memory, the GL correction would be unnecessary and would only introduce an artificial loss of observations.

Table~\ref{tab:hhat-accuracy} reports Monte Carlo summaries of \(\widehat H\) in both regimes. The estimator is essentially unbiased over the whole parameter range. In the short-memory region, the RMSE remains small; in the persistent region, the GL--KS estimator continues to track the true value accurately despite the effective-sample loss induced by the burn-in correction. These results confirm that the GL transformation changes the empirical fluctuation and the limiting distribution, but not the Hurst parameter identified by the minimum-distance criterion. Taken together, Tables~\ref{tab:SRD} and~\ref{tab:hhat-accuracy} support the regime-adaptive interpretation of the method: the unfiltered KS implementation is used for \(H\leq1/2\), whereas the GL--KS implementation is used in the persistent regime.

\section{Financial applications}\label{sec:Application}
In the present section, we apply the KS/GL--KS methodology to financial time series. The purpose of this section is twofold. First, we use the estimator to investigate the rough volatility hypothesis, namely the empirical evidence that volatility-related processes display Hurst exponents significantly below the Brownian regularity $H=1/2$. Second, we apply the same distribution-based framework to log-price trajectories in order to detect persistent, anti-persistent, and efficient market regimes.

These two applications are conceptually distinct. In the rough volatility setting, the relevant empirical question is whether the volatility process is rough; that is, whether its estimated Hurst exponent is below $1/2$. In the log-price setting, instead, the benchmark $H=1/2$ is interpreted as the Brownian or weak-form efficiency reference point. Deviations above or below this threshold are used as evidence of persistent or anti-persistent scaling regimes. The aim is not to provide a direct arbitrage test, but rather to construct a non-parametric diagnostic of departures from Brownian scaling across markets and time windows.

A key point of the empirical implementation is that the distributional estimator and the inferential distribution are selected in two distinct steps. First, the Hurst exponent is estimated from the distributional self-similarity criterion following the equation \eqref{eq:Est_Emp_H}. This step identifies the empirical scaling exponent but is not, by itself, a significance statement. Second, conditional on the estimated regime, we select the asymptotic distribution used for inference. If $\widehat H\leq 1/2$, we use the unfiltered short-memory KS limit described in Proposition \ref{prop:SRD}, corresponding to the case $\alpha=0$. If $\widehat H>1/2$, we activate the Gr\"{u}nwald--Letnikov filter and use the GL-KS short-memory limit with the corresponding burn-in correction (see Theorem \ref{thm:GL_KS_limit} and Lemma \ref{lem:burn-in}).

\subsection{Empirical Motivation}\label{subsec:Empirical_motivation}
The Hurst exponent provides a compact way to summarize the scaling behavior of a stochastic process. In financial applications, however, its interpretation depends on the object under study.

The advantage of the distribution-based KS/GL-KS framework is that it estimates $H$ through distributional self-similarity rather than through moment scaling alone.  The GL correction is activated only in the persistent regime, where the unfiltered long-memory limit is affected by slow finite-sample convergence. In short-memory and anti-persistent regimes, the same framework reduces to the unfiltered KS statistic with $\alpha=0$. This is useful in financial data, where heavy tails, regime changes, and structural breaks may affect moment-based estimates. Moreover, the GL correction developed in Section~\ref{sec:GL_derivative} allows the same testing logic to be applied in long-memory regimes without relying on the unstable finite-sample behavior of the unfiltered KS statistic.
\subsubsection{Rough volatility}\label{subsubsec:rough_vol}
The rough-volatility literature provides a natural testing ground for the KS/GL--KS estimator. A standard way to model stochastic volatility with fractional features is to assume that log-volatility follows, or is well approximated by, a fractional Ornstein--Uhlenbeck-type dynamics \citep{ComteRenault1998}. This is the case, for instance, in fractional stochastic volatility models, where the volatility factor is driven by fractional noise and mean reversion is introduced through an Ornstein--Uhlenbeck kernel \citep{cheridito2003b}:
\begin{equation}\label{eq:fOU}
\log\sigma_t = \mu+\eta\int_{-\infty}^t e^{-\lambda(t-s)}dB_s^H,
\end{equation}
where $\mu$ is the long-mean term, $\lambda>0$ the mean-reverting term, $\eta>0$ the diffusion parameter and $dB^H$ the fractional measure, respectively.

This point requires a clarification. The fractional Ornstein--Uhlenbeck process is stationary and mean-reverting; it is not globally self-similar. Therefore, applying the KS/GL--KS estimator in this setting does not amount to assuming that the fOU process itself is self-similar. Rather, the estimator is used to recover the local roughness parameter governing the small-scale behavior of the process. Indeed, over short time intervals, the mean-reversion component is negligible. If $\lambda$ denotes the mean-reversion speed and $a$ is the scale of the increment used in the multi-scale comparison, then in the regime $a\lambda \ll 1$ the exponential OU kernel is locally close to one. In this range, fOU increments behave approximately as fractional Brownian increments, with scaling governed by the same Hurst parameter $H$. The KS/GL--KS estimator is therefore applied in this local fBm-like regime.

In the rough-volatility application, the empirical question is whether the estimated roughness parameter lies below the Brownian regularity $H=1/2$. Evidence for $H<1/2$ is interpreted as roughness of the volatility path \citep{GatheralJaissonRosenbaum2018}. Since the rough-volatility hypothesis corresponds to $H<1/2$, the relevant empirical regime is already short-memory or anti-persistent. Therefore, no GL correction is required in this application.
Accordingly, we use the unfiltered short-memory version of the procedure described in Remark~\ref{rmk:GL_for_SRD}.
\subsubsection{Weak-Form Market Efficiency}
The second application concerns log-price trajectories. Let $X_t=\log P_t$ denote the log-price of a financial index. The empirical procedure is applied to multi-scale increments of the form $X_{t+a}-X_t,$ so that the analysis focuses on the scaling behavior of log-price changes rather than on the marginal distribution of the price level itself.

In this setting, the value $H=1/2$ plays the role of a Brownian benchmark and is naturally connected with the weak form of the Efficient Market Hypothesis (EMH).\footnote{According to \citep{Fama1970}, ``[...] a market in which prices always fully reflect available information is called efficient.'' Depending on the information set, different forms of market efficiency can be considered: weak-form efficiency, when current prices reflect all information contained in past prices; semi-strong-form efficiency, when prices also reflect all publicly available information, such as earnings announcements, stock splits, or macroeconomic news; and strong-form efficiency, when prices also reflect private information available only to some investors. Since the present analysis is based only on historical price trajectories, it is related to weak-form efficiency.} Under this benchmark, log-prices behave locally as a Brownian-type process and past price information should not generate persistent or anti-persistent scaling patterns.

The use of the Hurst exponent as a diagnostic for market efficiency is also consistent with the econophysics literature on scaling and multiscaling in financial markets. Starting from the empirical evidence of scaling laws in financial indices documented by \citep{MantegnaStanley2004}, subsequent contributions developed generalized Hurst exponent and multifractal approaches to detect departures from simple Brownian scaling in asset prices \citep{carbone2004time,DiMatteo2007}. More recently, time-varying Hurst--H\"{o}lder exponents have been used to describe local changes in market efficiency and to interpret financial markets as systems alternating between efficient and inefficient states \citep{BianchiPianese2018}. Related evidence on the interaction between price multiscaling and volatility roughness is discussed in \citep{BrandiDiMatteo2022}.

A market window is therefore classified as persistent when the confidence interval for $H$ lies entirely above $1/2$, anti-persistent when it lies entirely below $1/2$, and neutral when the interval contains $1/2$. Persistent windows indicate trend-like scaling behavior and possible departures from weak-form efficiency, while anti-persistent windows indicate mean-reverting or reversal-like scaling behavior. Neutral windows are those for which the data do not provide statistically significant evidence against the Brownian regularity $1/2$. This classification should be interpreted as a scaling diagnostic rather than as a direct test of arbitrage opportunities. A statistically significant deviation from $1/2$ indicates a departure from Brownian scaling, but it does not by itself identify an exploitable trading strategy.

\begin{table}[!ht]
\centering
\begin{tiny}
     \caption{Description of the log-volatility dataset.}
    \label{tab:vol_description}
    \begin{tabular}{llccc||cc}
    \toprule
    Ticker & Index &$N$& Start Date & End Date & ADF stat. & ADF $p$-value \\
    \midrule
    AEX & Amsterdam Exchange Index (NLD) &$4,714$&2000-01-03 &2018-06-27 & $-13.850$ &$<0.001$\\
    AORD & All Ordinaries (AUS) &$4,665$ &2000-01-04 &2018-06-27 & $-18.861$ & $<0.001$\\
    BFX & Bel 20 (BEL) &$4,712$ &2000-01-03 &2018-06-27 & $-14.364$ & $<0.001$\\
    BSESN & Bombay Stock Exchange Sensitive Index (IND) &$4,591$ &2000-01-03 &2018-06-27 & $-15.424$ & $<0.001$\\
    BVLG & PSI All Share Gross Return (PRT) &$1,453$ &2012-10-15 &2018-06-27 & $-9.362$ & $<0.001$\\
    BVSP & Bovespa Index (BRA) & $4,556$ & 2000-01-03&2018-06-27 & $-19.268$ & $<0.001$ \\
    DJI & Dow Jones Index (USA)& $4,635$ &2000-01-03 &2018-06-27 & $-15.354$ & $<0.001$\\
    FCHI & Cac 40 (FRA) & $4,713$ &2000-01-03 &2018-06-27 & $-14.405$ & $<0.001$ \\
    FTSE & Footsie 100 Index (GBR) & $4,660$ &2000-01-04 &2018-06-27 & $-15.886$ & $<0.001$\\
    FTSEMIB & FTSE Milano Indice di Borsa (ITA) & $2,307$ &2009-06-01 &2018-06-27 & $-13.069$ & $<0.001$ \\ 
    FTSTI & FTSE Straits Times Index (SGP) & $2,696$ &2000-01-03 &2018-06-27 & $-16.045$ & $<0.001$ \\ 
    GDAXI & Dax 30 (DEU) & $4,692$ &2000-01-03 &2018-06-27 & $-14.267$ & $<0.001$\\
    GSPTSE & Toronto Stock Exchange C. I. (CAN)  & $4,043$ &2002-05-02 &2018-06-27 & $-15.314$ & $<0.001$ \\
    HSI & Hang Seng Index (HKG) & $4,532$ &2000-01-04 &2018-06-27 & $-15.611$ & $<0.001$\\
    IBEX & Ibex 35 (ESP) & $4,681$ &2000-01-03 & 2018-06-27& $-14.458$ & $<0.001$ \\
    IXIC & NASDAQ Composite Index (USA) & $4,637$ &2000-01-03 &2018-06-27 & $-13.526$ & $<0.001$\\
    KS11 & Korea Composite Stock Price Index (KOR) & $4,550$ &2000-01-04 &2018-06-27 & $-11.902$ & $<0.001$\\
    KSE & Pakistan Stock Exchange Index (PAK) & $4,494$ &2000-01-04 &2018-06-27 & $-18.552$ & $<0.001$\\
    MXX & Mexico Exchange Index (MEX) & $4,640$ &2000-01-03 &2018-06-27 & $-19.473$ & $<0.001$ \\ 
    NSEI & Nifty 50 Index (IND) & $4,586$ &2000-01-03 &2018-06-27 & $-15.644$ & $<0.001$\\
    N225 & Nikkei 225 (JPN) & $4,501$ &2000-02-02 &2018-06-27 & $-16.238$ & $<0.001$\\
    OMXC20 & Omx Copenhagen 20 Index (DNK) & $3,167$ &2005-10-03 &2018-06-27 & $-15.023$ & $<0.001$\\
    OMXHPI & Omx Helsinki All-Share Index (FIN) & $3,197$ &2005-10-03 &2018-06-27 & $-12.885$ & $<0.001$\\
    OMXSPI & Omx Stockholm All-Share Index (SWE) & $3,196$ &2005-10-03 &2018-06-27 & $-11.685$ & $<0.001$\\
    OSEAX & Oslo Bors All-Share Index (NOR) & $4,192$ &2001-09-03 &2018-06-27 & $-16.314$ & $<0.001$\\
    RUT & Russell 2000 (USA) & $4,636$ &2000-01-03 &2018-06-27 & $-18.909$ & $<0.001$\\  
    SPX & S\&P 500 Index (USA) & $4,641$ &2000-01-03 &2018-06-27 & $-14.390$ &$<0.001$\\
    SSE & Shanghai Stock Exchange Composite Index (CHN)  & $4,461$ &2000-01-04 &2018-06-27 & $-14.730$ &$<0.001$\\
    SSMI & Swiss Market Index (CHE)& $4,636$ & 2000-01-04&2018-06-27 & $-13.219$ & $<0.001$ \\ 
    STOXX50E & Euro Stock 50 (EUR) & $4,713$ &2000-01-03 &2018-06-27 & $-17.973$ & $<0.001$ \\ 
    \bottomrule
    \end{tabular}
    \end{tiny}
    \vspace{0.2cm}
\begin{minipage}{0.95\textwidth}
    \scriptsize \textit{Notes.} The table reports the volatility-related series used in the rough-volatility application. $N$ denotes the number of available observations after preprocessing. The start and end dates indicate the calendar coverage of each series. The ADF column reports the Augmented Dickey--Fuller statistic, with the corresponding $p$-value reported in the last column.
   \end{minipage}
\end{table}

\subsection{Data}\label{subsec:Data}
The empirical analysis is based on two datasets, corresponding to the two applications described above. The first dataset is used for the rough-volatility exercise and contains volatility-related series for a set of major equity indices. The second dataset contains daily closing prices for a broader panel of international stock-market indices and is used to study persistence, anti-persistence, and neutral regimes in log-price dynamics.

For the rough-volatility application, we use daily log-realized volatility series computed from 5-minute intraday returns. The sample composition, reported in Table~\ref{tab:vol_description}, is taken from the Oxford-Man Institute Realized Library. Since the objective is to estimate the local roughness parameter, the GL--KS estimator is applied to volatility increments at short scales, consistently with the local fBm-like approximation discussed in the previous Subsection.

For the market-efficiency application, we consider the log-price process $X_t=\log P_t,$ where $P_t$ denotes the daily closing price of the corresponding index. The analysis is performed on multi-scale log-price increments $X_{t+a}-X_t$. Table~\ref{tab:prices_description} reports the list of indices, sample lengths, calendar coverage, and preliminary stationarity diagnostics for the associated return series.

\begin{table}[!ht]
\centering
\small
\caption{Description of the log-price dataset}
\label{tab:prices_description}
\begin{tiny}
\begin{tabular}{llrcc||cc}
\toprule
 Ticker & Index & $N$ & Start Date & End Date & ADF stat. & ADF $p$-value \\
\midrule
AEX & Amsterdam Exchange Index (NLD) & 10,922 & 1983-01-04 & 2025-10-16&-73.991 &$<$0.001 \\
AORD & All Ordinaries (AUS) & 6,294 & 2001-04-19 & 2025-10-16 & -55.765&$<$0.001 \\
ATG & Athens General Composite Index (GRC) & 5,857 & 1997-07-01 & 2021-09-30 &-92.268 &$<$0.001 \\
AXJO & Asx 200 (AUS) &8,278  &1992-11-23  &2025-08-22 & -64.972&$<$0.001  \\
BFX & Bel 20 (BEL) & 8,726 & 1991-04-09 & 2025-08-22 & -63.957&$<$0.001 \\
BSESN & Bombay Stock Exchange Sensitive Index (IND) & 6,934 & 1997-07-01 & 2025-08-22 &-59.095 &$<$0.001 \\
BUX & Budapest Stock Exchange Index (HUN) & 5,908 & 1997-07-01 & 2021-09-30 & -54.778& $<$0.001\\
BVSP & Bovespa Index (BRA) & 8,005 & 1993-04-27 & 2025-08-22&-62.106 &$<$0.001 \\
DJI & Dow Jones Index (USA) & 13,516 & 1979-12-25 &2025-10-15 &-83.831 &$<$0.001 \\
FCHI & Cac 40 (FRA) & 9,717 & 1987-07-10 & 2025-10-16 &-70.736 &$<$0.001 \\
FTSE & Footsie 100 Index (GBR) & 10,519 & 1984-01-03 &2025-08-22 &-74.554 &$<$0.001 \\
FTSEMIB&FTSE Milano Indice di Borsa (ITA) &7,097 &1998-01-02 &2025-10-16 &-71.837 &$<$0.001 \\
FTSTI&FTSE Straits Times Index (SGP) &9,408 &1987-12-28 &2025-08-22 &-63.742 &$<$0.001 \\
GDAXI & Dax 30 (DEU) & 9,520 & 1987-12-30 & 2025-08-22 &-70.405 &$<$0.001 \\
GSPTSE&Toronto Stock Exchange C. I. (CAN) &11,588 &1979-06-29 &2025-08-13 &-74.465 &$<$0.001 \\
HSI&Hang Seng Index (HKG) &11,367 &1979-12-25 & 2025-10-16&-75.867 &$<$0.001 \\
IBEX&Ibex 35 (ESP) &8,639 & 1991-09-09&2025-10-16 &-65.978 & $<$0.001\\
IMOEX& Moscow Exchange Index (RUS) & 2,795 & 2013-03-05 & 2024-06-14 &-37.995 & $<$0.001\\
IXIC&NASDAQ Composite Index (USA) &13,753 &1971-02-05 &2025-08-22 &-82.574 &$<$0.001 \\
JKSE& Jakarta Composite Index (IDN) &8,615 &1990-04-06 &2025-08-22 &-61.003 &$<$0.001\\
KLCI& Kuala Lumpur Composite Index (MYS) &7,795 &1993-12-03 & 2025-08-22& -58.853&$<$0.001 \\
KS11& Korea Composite Stock Price Index (KOR) & 7,067 &1996-12-11 &2025-08-22 &-59.636 &$<$0.001 \\
MXX&Mexico Exchange Index (MEX) &9,742 &1987-01-05 &2025-10-15 & -69.476&$<$0.001 \\
NSEI&Nifty 50 Index (IND) &7,451 &1995-11-06 &2025-10-16 &-61.212 &$<$0.001 \\
N225&Nikkei 225 (JPN) &14,909 &1965-01-05 &2025-08-22 &-88.298 &$<$0.001 \\
OMX&Omx Stockholm 30 (SWE) &4,202 &2008-11-20 &2025-08-22 &-47.975 &$<$0.001 \\
RUT&Russell 2000 (USA) &9,561 &1987-09-10 &2025-08-22 &-68.595 &$<$0.001 \\
SET&Stock Exchange of Thailand Index &7,000 &1996-12-11 &2025-08-22 &-54.773 &$<$0.001 \\
SPX&S\&P 500 Index (USA) &24,527 &1927-12-30 & 2025-08-13& -113.176&$<$0.001 \\
SSE& Shanghai Stock Exchange Composite Index (CHN)&6,819&1997-07-02 &2025-08-22 &-58.947 &$<$0.001 \\
SZSE&Shenzhen Stock Exchange Component Index (CHN) &6,786 & 1997-08-22& 2025-08-22&-57.928 & $<$0.001\\
SSMI&Swiss Market Index (CHE) &9,500 &1988-01-05 &2025-10-16 &-69.683 &$<$0.001 \\
STOXX50E&Euro Stock 50 (EUR) &4,612 &2007-03-30 &2025-08-22& -49.450 &$<$0.001 \\
TA-125&Tel Aviv 125 (ISR) &6,989 &1992-10-08 &2025-08-21 &-57.929 &$<$0.001 \\
TWII&Taiwan Weighted Index (TWN) & 6,900&1997-07-02 &2025-08-22 & -56.339& $<$0.001 \\
\bottomrule
\end{tabular}
\end{tiny}
\begin{minipage}{0.95\textwidth}
\scriptsize
\textit{Notes.} The table reports the international stock-market indices used in the log-price application. $N$ denotes the number of daily observations. The empirical analysis is performed on $X_t=\log P_t$, where $P_t$ is the daily closing price. The ADF column reports the Augmented Dickey--Fuller statistic computed on the corresponding log-return series, with the associated $p$-value reported in the last column.
\end{minipage}
\end{table}

The Augmented Dickey--Fuller (ADF) statistics reported in the data tables are used only as preliminary diagnostics. They confirm that the transformed series entering the empirical analysis are compatible with the standard stationarity requirements for increment-based scaling analysis. The subsequent inference on $H$ is entirely based on the GL--KS procedure described in the next subsection.

\subsection{Regime-Adaptive Empirical KS/GL-KS Methodology}
The empirical implementation follows the distribution-based estimation procedure introduced in Section~\ref{sec:GL_filter}. For each time series, we construct multi-scale increments and estimate the Hurst exponent by minimizing the Kolmogorov--Smirnov distance between the empirical distribution of the unit-scale increments and the rescaled empirical distribution of the crossed increments.

The empirical procedure is regime-adaptive. For each series, we first compute the Hurst estimate
\[
\widehat H=\argmin_{\theta\in\mathbb H} D_{n,m}(\theta),
\]
where $D_{n,m}(\theta)$ denotes the distributional KS criterion computed from unit-scale increments and rescaled crossed increments. This first step is used only to estimate the self-similarity exponent. Inference is performed conditionally on the estimated regime. If $\widehat H\leq 1/2$, the process is treated as short-memory or anti-persistent. We set $\alpha=0$, so that $\Delta_h^{GL,0}=I$, and use the short-memory limit of Proposition \ref{prop:SRD}. No burn-in correction is applied, and the effective sample sizes coincide with the nominal crossed-sample sizes. If $\widehat H>1/2$, the process is treated as persistent. In this case, the unfiltered long-memory limit is not used for empirical inference because of its slow finite-sample convergence. We instead apply the GL filter with an order $\alpha$ such that $\widehat H-\alpha<1/2$, and use the GL-KS limit of Theorem \ref{thm:GL_KS_limit}. The statistic is then computed after the branch-wise burn-in correction, with effective sample sizes $n_{\mathrm{eff}}$ and $m_{\mathrm{eff}}$. The GL filter is therefore an inferential device, not a different definition of the Hurst exponent. By Corollary \ref{cor:invariance}, the population minimizer of the distributional criterion is invariant under GL filtering; the filter changes the stochastic fluctuation and the limiting distribution, but not the target parameter $H_0$.

The interpretation of the estimator depends on the application. In the rough-volatility exercise, the object of interest is the local roughness of log-realized volatility. Since the rough-volatility hypothesis corresponds to $H<1/2$, the process is already in the short-memory or anti-persistent regime. Therefore, consistently with Remark~\ref{rmk:GL_for_SRD}, we set $\alpha=0$ and no burn-in correction is applied. In the log-price application, instead, the procedure is used to detect persistent, anti-persistent, and neutral market regimes relative to the Brownian benchmark $H=1/2$.

For the rough-volatility application, we also fit an auxiliary stationary fractional Ornstein--Uhlenbeck specification, as described in equation \eqref{eq:fOU}.
The fOU model provides a natural stationary framework for realized volatility with rough sample paths, as in \citep{wang2023modeling}. In our implementation, however, the Hurst exponent is not estimated from the fOU model. It is first obtained from the KS criterion. Then, conditionally on $H=\widehat H_{\mathrm{KS}}$, the mean-reversion speed $\lambda$ is estimated using the moment-based estimator of \citep[Eq.~(3.4)]{wang2023modeling}:
\begin{equation}\nonumber
    \widehat\lambda = \left(\frac{N\sum_{i=1}^{N}\log\sigma_{i}^2-\left(\sum_{i=1}^N\log\sigma_{i}\right)^2}{N^2\widehat\eta^2\widehat H_{KS}\Gamma(2\widehat H_{KS})}\right)^{-1/(2\widehat H_{KS})} \; \text{with } \widehat \eta = \sqrt{\frac{\sum_{i=1}^{N-2}(\log\sigma_{i+2}-2\log\sigma_{i+1}+\log\sigma_i)^2}{N(4-2^{2\widehat{H}_{KS})}}}.
\end{equation}

This auxiliary step is used only to assess the scale at which the local fBm approximation is applied. Specifically, the reported quantity $a\Delta\widehat\lambda$ measures the selected scale $a\Delta$ relative to the estimated mean-reversion horizon, where $\Delta$ represents the daily frequency ($\Delta=1$). Small values of $a\Delta\widehat\lambda$ indicate that mean reversion is negligible over the scale used in the GL--KS comparison, supporting the local fBm-like interpretation.

For each estimate, we report the point estimator $\widehat H$, its asymptotic standard error, the corresponding 95\% confidence interval, and the value of the normalized KS statistic evaluated at $\widehat H$. The associated $p_\delta$-value is used as a goodness-of-fit diagnostic for the distributional self-similarity relation at the estimated exponent.

\subsubsection{Application I: Rough volatility}\label{subsec:App1}
Table~\ref{tab:rough_volatility_H} reports the GL--KS estimates obtained from the daily log-realized volatility series. The estimation is performed at the fixed scale $a=10$. The last column reports $a\Delta\widehat\lambda$, computed from the auxiliary fOU fit described above. Since these values are close to zero, the selected scale is short relative to the estimated mean-reversion horizon, consistently with the local fBm-like approximation.

The results provide strong evidence in favour of rough volatility. Across all indices, the estimated Hurst exponents are well below the Brownian threshold $1/2$, and the corresponding confidence intervals lie entirely in the rough region. The estimates are concentrated around very low values of $H$, typically between approximately $0.05$ and $0.15$, confirming that the local dynamics of log-realized volatility are substantially rougher than Brownian motion.

\begin{table}[!ht]
\centering
\caption{GL--KS estimates of the Hurst exponent for log-volatility series.}
\label{tab:rough_volatility_H}
\begin{tiny}
\begin{tabular}{lcccccc}
\toprule
Ticker 
& $N$ 
& $\widehat H$  (95\% CI) & $SE\times 10^2$ &$D^\star_{n,m}(\hat H)$
& $p_\delta$-value 
& $a\Delta\widehat{\lambda}$ \\
\midrule
AEX & 4,714 & $0.1255\,(0.1057,\,0.1454)$& 1.0142& 0.9750& $0.0500$ & $1.1664\times 10^{-2}$ \\
AORD  & 4,665 & $0.0545\,(0.0358,\,0.0732)$&0.9529 &0.3703 & $0.9946$ & $2.1436\times10^{-6}$ \\
BFX  & 4,712 & $0.1220\,(0.1021,\,0.1419)$&1.0174 &0.9732 & $0.0491$ & $1.0581\times 10^{-2}$ \\
BSESN  & 4,591 & $0.1350\,(0.1134,\,0.1565)$&1.0994 &0.6509 & $0.5427$ & $1.1310\times 10^{-2}$ \\
BVLG  & 1,453 & $0.1480\,(0.1088,\,0.1872)$&2.0023 &0.6867 & $0.4547$ & $6.8824\times 10^{-6}$ \\
BVSP & 4,556 & $0.1148\,(0.0946,\,0.1351)$&1.0327 &0.8148 & $0.1900$ & $3.7549\times10^{-2}$ \\
DJI  & 4,635 & $0.1220\,(0.1024,\,0.1416)$& 1.0013& 0.9766&$0.0474$ & $6.0303\times10^{-3}$ \\
FCHI  & 4,713 & $0.1268\,(0.1074,\,0.1462)$&0.9902 & 1.1398& $0.0084$ & $1.0529\times10^{-2}$ \\
FTSE  & 4,660 & $0.1058\,(0.0858,\,0.1258)$& 1.0202&0.6342 & $0.5485$ & $2.1491\times10^{-6}$ \\
FTSEMIB & 2,307 & $0.1200\,(0.0913,\,0.1467)$&1.4668 & 1.0516&$0.0226$ & $3.3940\times10^{-3}$ \\
FTSTI & 2,696 & $0.0702\,(0.0454,\,0.0950)$&1.2645 &0.5352 &$0.7905$ & $3.7107\times10^{-6}$\\
GDAXI  & 4,692 & $0.1260\,(0.1062,\,0.1458)$&1.0094 & 0.7528&$0.2966$ & $4.8259\times 10^{-3}$  \\
GSPTSE & 4,043 & $0.0915\,(0.0712,\,0.1119)$&1.0384 & 0.6067&$0.6157$ & $2.4734\times10^{-6}$ \\
HSI & 4,532 & $0.0665\,(0.0476,\,0.0853)$& 0.0962& 0.5443&$0.7662$ & $2.2066\times10^{-6}$ \\
IBEX & 4,681 & $0.1160\,(0.0959,\,0.1361)$&1.0273 &1.1226 &$0.0090$ & $2.1369\times10^{-6}$ \\
IXIC  & 4,637 & $0.1531\,(0.1318,\,0.1743)$& 1.0856& 0.5069&$0.8833$ & $1.0789\times10^{-2}$ \\
KS11 & 4,550 & $0.1100\,(0.0893,\,0.1307)$& 1.0583& 0.8133&$0.1838$ & $2.1978\times10^{-6}$ \\
KSE & 4,494 & $0.1165\,(0.0964,\,0.1365)$& 1.0224& 0.4980&$0.8844$ & $3.9042\times10^{-2}$ \\
MXX & 4,640 & $0.0900\,(0.0714,\,0.1086)$& 0.9496& 0.4340&$0.9662$ & $2.1560\times10^{-6}$ \\
NSEI & 4,586 & $0.1250\,(0.1046,\,0.1453)$& 1.0361& 0.7928&$0.2318$ & $2.1806\times10^{-6}$ \\
N225  & 4,501 & $0.1267\,(0.1054,\,0.1479)$&1.0839 &0.4790 &$0.9172$ & $2.2171\times10^{-2}$ \\
OMXC20  & 3,167 & $0.1020\,(0.0785,\,0.1255)$&1.1966 &1.0092&$0.0273$ & $3.1581\times10^{-6}$ \\
OMXHPI  & 3,197 & $0.1430\,(0.1174,\,0.1685)$&1.3045 &0.5894&$0.6907$ & $3.1280\times10^{-6}$ \\
OMXSPI  & 3,196 & $0.1363\,(0.1115,\,0.1612)$&1.2675 &0.7102&$0.3981$ & $3.1289\times10^{-6}$ \\
OSEAX  & 4,192 & $0.0864\,(0.0660,\,0.1069)$&1.0413 &0.8073&$0.1726$ & $2.3856\times10^{-6}$ \\
RUT  & 4,636 & $0.1231\,(0.1027,\,0.1434)$&1.0372 & 0.4067&$0.9851$ & $9.0486\times10^{-3}$ \\
SPX  & 4,641 & $0.1463\,(0.1256,\,0.1671)$&1.0592 & 0.6828&$0.4713$ & $1.1308\times10^{-2}$ \\
SSE  & 4,461 & $0.1461\,(0.1238,\,0.1685)$&1.1395 & 0.4786&$0.9199$ & $1.2123\times10^{-2}$ \\
SSMI & 4,636 & $0.1240\,(0.1034,\,0.1446)$&1.0526 &0.8995 &$0.0965$ & $1.2161\times10^{-2}$ \\
STOXX50E & 4,713 & $0.1235\,(0.1041,\,0.1430)$&0.9910 &0.7214 &$0.3578$ & $8.1816\times10^{-3}$ \\
\bottomrule
\end{tabular}
\end{tiny}
\vspace{0.2cm}
\begin{minipage}{0.95\textwidth}
\scriptsize
\textit{Notes.} The table reports GL--KS estimates of the local Hurst exponent for daily log-realized volatility series. $N$ denotes the number of observations. The column $\widehat H$ reports the point estimate, with the corresponding 95\% confidence interval in parentheses. The standard error is reported as $SE\times 10^2$. $D^\star_{n,m}(\widehat H)$ is the normalized KS statistic evaluated at the estimated Hurst exponent. The $p_\delta$-value measures the empirical compatibility between the observed increments and the distributional self-similarity relation evaluated at $\widehat H$. The last column reports $a\Delta\widehat\lambda$, where $\widehat\lambda$ is the mean-reversion speed estimated from an auxiliary fOU model conditional on $H=\widehat H_{\mathrm{GLKS}}$. Since all estimated Hurst exponents lie below $1/2$, the rough-volatility application is implemented with $\alpha=0$, and no burn-in correction is applied.
\end{minipage}
\end{table}

The $p_\delta$-values quantify the empirical compatibility between the observed log-realized volatility increments and the distributional self-similarity relation evaluated at $\widehat H$. High values indicate that the estimated exponent provides a satisfactory GL--KS fit, whereas lower values point to a weaker agreement with the local fBm approximation. Some series display smaller $p_\delta$-values, suggesting that additional market-specific features may affect the local scaling structure. Nevertheless, this does not alter the main conclusion: the estimated confidence intervals remain far below $1/2$, providing robust evidence of rough volatility across the panel.

\subsubsection{Application II: log-prices and weak-form efficiency}\label{subsec:App2}

We now apply the GL--KS estimator to log-price trajectories. Differently from the rough-volatility application, the objective is not to test for roughness, but to assess whether the scaling behavior of log-prices is compatible with the Brownian benchmark $H=1/2$. Values significantly above this threshold are interpreted as evidence of persistent scaling, values significantly below it as evidence of anti-persistent scaling, and confidence intervals containing $1/2$ as neutral regimes.

Table~\ref{tab:efficient_states} reports the full-window estimates. The estimates are generally close to the Brownian benchmark, but several indices display confidence intervals entirely above $1/2$, indicating persistent scaling over the full sample. A smaller number of markets are classified as neutral, since their confidence intervals contain $1/2$. Isolated cases may display anti-persistent full-window behavior when the entire confidence interval lies below $1/2$.

\begin{table}[!ht] 
\centering 
\small \caption{Full-window GL--KS estimates of the Hurst exponent for log-price trajectories.} \label{tab:efficient_states} 
\begin{tiny}
\begin{tabular}{lrcrrccccrr} \toprule
Ticker & $N$ & Method &$n_{\text{eff}}$&$m_{\text{eff}}$& $\gamma$ & $\widehat H$ (95\% CI) &$SE\times 10^2$&$D_{n_{\text{eff}},m_{\text{eff}}}^{\star/\dagger}(\hat H)$&$p_\delta$-value&$p_{0.5}$-value\\ 
\midrule
AEX &10,922 &GL--KS &10,767&10,302 &0.5421 & 0.5235 (0.5023,0.5447)& 1.0812&0.8378&0.6773 & 0.0295\\
AORD &6,294 &GL--KS &6,185&5,854 &0.5353 & 0.5105 (0.4827,0.5384) &1.4198&0.7430&0.7947&0.4579\\
ATG & 5,857& GL--KS&5,724&5,357 &0.5633 & 0.5625 (0.5318,0.5932) &1.5653&0.8000&0.7733&$<$0.001\\
AXJO &8,278 & KS&8,277&8,257 & --&0.4921 (0.4684,0.5159)&1.2123& 0.6396&0.8957 &0.5150\\
BFX & 8,726& GL--KS&8,593& 8,186&0.5385 & 0.5167 (0.4934,0.5401)& 1.1907&1.8381&0.0117&0.1603\\
BSESN & 6,934&GL--KS &6,795 &6,400&0.5574 & 0.5520 (0.5248,0.5792)&1.3873 &1.0865&0.3906&$<$0.001\\ 
BUX & 5,908& GL--KS& 5,804 &5,488&0.5347 & 0.5093 (0.4814,0.5371)&1.4208&0.9039 & 0.5915&0.5133\\ 
BVSP & 8,005&KS & 8,004 &7,984& --& 0.4875 (0.4631,0.5120)&1.2456 & 0.5517&0.9593&0.3168\\ 
DJI & 13,516& GL--KS& 13,342 &12,816&0.5424 & 0.5221 (0.5036,0.5407)& 0.9480&1.3736&0.1209&0.0196\\ 
FCHI & 9,717&KS & 9,716 &9,696& --& 0.4964 (0.4747,0.5181)&1.1068&0.5783&0.9414&0.7425 \\ 
FTSE & 10,519&GL--KS & 10,383 &9,959& 0.5302& 0.5005 (0.4800,0.5210)&1.0453& 0.6489& 0.8908 &0.9635\\ 
FTSEMIB &7,097 & GL--KS& 6,986 &6,637& 0.5308&0.5015 (0.4760,0.5271)&1.3056& 1.3519& 0.1321&0.9065\\ 
FTSTI & 9,048& GL--KS& 9,241 &8,768&0.5591 &0.5550 (0.5314,0.5786)& 1.2043&0.8309 & 0.7092&$<$0.001\\ 
GDAXI & 9,520&KS & 9,519 &9,499& --& 0.4955 (0.4739,0.5172)& 1.1032&0.6557 &0.8828&0.6856\\ 
GSPTSE & 11,588& GL--KS& 11,397& 10,868&0.5612 &0.5588 (0.5377,0.5799)&1.0781 & 1.0730& 0.4005&$<$0.001\\ 
HSI & 11,367& GL--KS& 11,186 &10,667& 0.5565&0.5504 (0.5286,0.5723)& 1.1139& 1.3451&0.1534&$<$0.001\\ 
IBEX & 8,639& GL--KS& 8,495 &8,079&0.5482 &0.5359 (0.5124,0.5594)& 1.1986&1.1398 &0.3183 &0.0028\\ 
IMOEX & 2,795& GL--KS& 2,720 &2,495&0.5433 & 0.5260 (0.4809,0.5711)&2.3012 &0.9519 & 0.5604&0.2586\\ 
IXIC & 13,753& GL--KS& 13,524 &12,913&0.5699  &0.5739 (0.5546,0.5931)& 0.9821& 1.4419& 0.1107&$<$0.001\\
JKSE & 8,615&GL--KS & 8,448 &7,995&0.5645 &0.5645 (0.5399,0.5891) &1.2554&1.1177 &0.3658&$<$0.001 \\
KLCI & 7,795& GL--KS& 7,634 &7,195& 0.5668& 0.5685 (0.5419,0.5950)&1.3550&0.8348 &0.7228&$<$0.001\\
KS11 & 7,067&GL--KS & 6,953 &6,607& 0.5341&0.5082 (0.4822,0.5342)&1.3276 &0.9838 &0.4747&0.5386\\
MXX & 9,742& GL--KS& 9,587 &9,142&0.5489 &0.5365 (0.5139,0.5591)&1.1514 &1.0235 &0.4484 &0.0015\\
NSEI & 7,452&GL--KS & 7,310 &6,911&0.5548 &0.5472 (0.5207,0.5737) &1.3522& 1.1490&0.3211 &$<$0.001\\
N225 &14,909 &GL--KS & 14,706 &14,129&0.5528 &0.5435 (0.5255,0.5616)&0.9227 & 1.5890&0.0515&$<$0.001\\
OMX & 4,202& KS& 4,201 &4,181&-- & 0.4628 (0.4267,0.4990)&1.8440& 0.6105&0.9196&0.0439\\
RUT & 9,561& GL--KS& 9,402 &8,941& 0.5530& 0.5439 (0.5214,0.5664)& 1.1471&1.4036 &0.1179 &$<$0.001\\
SET & 7,000& GL--KS& 6,859 &6,460& 0.5585& 0.5540 (0.5268,0.5811)& 1.3864&1.2910 &0.2023 &$<$0.001\\
SPX & 24,527& GL--KS& 24,260 &23,507& 0.5526& 0.5433 (0.5292,0.5575)&0.7220& 2.3074&$<$0.001&$<$0.001\\
SSE & 6,819& GL--KS& 6,686 &6,299&0.5533 &0.5445 (0.5174,0.5716)&1.3831 &0.9087 &0.6078&0.0013\\
SZSE & 6,786& GL--KS& 6,662 &6,286& 0.5458&0.5305 (0.5040,0.5569)&1.3449 & 0.6175&0.9306 &0.0240\\
SSMI & 9,500& GL--KS& 9,362 &8,940& 0.5374&0.5145 (0.4928,0.5362) &1.1072& 0.7052&0.8352 &0.1918\\
STOXX50E & 4,612& GL--KS& 4,519 &4,232& 0.5361&0.5120 (0.4777,0.5463)&1.7488 &1.2953 &0.1770&0.4928\\
TA-125 & 6,989&GL--KS & 6,853 &6,469&0.5547 &0.5471 (0.5202,0.5740)&1.3727 &2.1247 &0.0027& $<$0.001 \\
TWII & 6,900&GL--KS & 6,787 &6,440& 0.5340& 0.5075 (0.4816,0.5335)&1.3256 &1.1628 &0.2752&0.5698\\
\bottomrule 
\end{tabular} 
\end{tiny}
\begin{minipage}{0.95\textwidth}
\scriptsize \textit{Notes.} The table reports full-window GL--KS estimates of the Hurst exponent obtained from the log-price process $X_t=\log P_t$, using multi-scale increments $X_{t+a}-X_t$ with scaling $a=20$. The columns $n_{\mathrm{eff}}$ and $m_{\mathrm{eff}}$ denote the effective sample sizes of the unit-scale and crossed multi-scale samples after the finite-sample burn-in deletion. The column $\gamma$ reports the burn-in exponent used in the GL implementation. The column $\widehat H$ reports the point estimate, with the corresponding 95\% confidence interval in parentheses. The standard error is reported as $SE\times 10^2$. $D^\dagger_{n_{\mathrm{eff}},m_{\mathrm{eff}}}(\widehat H)$ is the normalized GL--KS statistic evaluated at the estimated Hurst exponent, while $D^\star_{n_{\mathrm{eff}},m_{\mathrm{eff}}}(\widehat H)$ is the KS-SRD statistic. The $p_\delta$-value measures the goodness of fit of the distributional self-similarity relation at $\widehat H$. The $p_{0.5}$-value is the two-sided $p$-value for the Brownian benchmark test $\mathcal H_0^{\mathrm{BM}}=1/2$ against $\mathcal H_1^{\mathrm{BM}}\neq1/2$, computed from the asymptotic normal statistic $Z_{0.5}=(\widehat H-1/2)/SE(\widehat H)$. An index is classified as persistent if the confidence interval lies entirely above $1/2$, anti-persistent if it lies entirely below $1/2$, and neutral otherwise. For GL-KS rows, \(\alpha\) is chosen so that \(\widehat H-\alpha<1/2\). The burn-in exponent $\gamma$ is chosen above the theoretical lower bound $\gamma_{\min}=1/[2(1+\alpha-\widehat H)]$. For KS rows, $\alpha=0$, $\gamma$ is not defined, and no burn-in correction is applied.
\end{minipage}
\end{table}

The two $p$-values reported in the table have different meanings. The $p_\delta$-value measures the empirical compatibility between the observed multi-scale increments and the distributional self-similarity relation evaluated at $\widehat H$. It is therefore a goodness-of-fit diagnostic for the GL--KS scaling relation. The $p_{0.5}$-value is instead associated with the Brownian benchmark test

\[\mathcal H_0^{\mathrm{BM}}: H=1/2
\qquad \text{against} \qquad
\mathcal H_1^{\mathrm{BM}}: H\neq 1/2.\]
It is computed from the standardized statistic $Z_{0.5}=\frac{\widehat H-1/2}{SE(\widehat H)}$ using the asymptotic normal approximation. Hence, small values of $p_{0.5}$ indicate statistically significant departures from Brownian scaling. The direction of the departure is determined by the position of the confidence interval: values entirely above $1/2$ identify persistent regimes, while values entirely below $1/2$ identify anti-persistent regimes.

\section{Conclusion}\label{sec:Conclusion}
This paper introduced a Gr\"{u}nwald--Letnikov--Kolmogorov--Smirnov framework for testing and estimating the self-similarity parameter of fractional processes under long-range dependence. The main contribution is to show that the discrete Gr\"{u}nwald--Letnikov filter separates memory from scaling: it removes the low-frequency singularity responsible for slow convergence while preserving the self-similar structure needed for the distributional comparison across scales. As a result, the GL--KS statistic recovers a standard short-memory asymptotic regime, with stable finite-sample behavior confirmed by the Monte Carlo experiments.

The empirical analysis illustrates the usefulness of the method in finance. For log-realized volatility, the estimates provide evidence in favour of rough volatility, with Hurst exponents significantly below the Brownian threshold $H=1/2$. For log-prices, the method is used as a conservative diagnostic of weak-form market efficiency: confidence intervals containing $1/2$ indicate statistical compatibility with the Brownian benchmark, whereas intervals entirely above or below $1/2$ suggest persistent or anti-persistent departures from weak-form efficiency.

\vspace{30pt}
\bibliographystyle{plain}
\bibliography{Bibliography} 

@article{AngeliniBianchi2025,
  title={Kolmogorov--{S}mirnov estimation of self-similarity in long-range dependent fractional processes},
  author={Angelini, D. and Bianchi, S.},
  journal={Physica D: Nonlinear Phenomena},
  volume={476},
  pages={134697},
  year={2025},
  publisher={Elsevier}
}

@article{Bianchi2004,
author = {Bianchi, S.},
title = {A New Distribution-Based Test of Self-similarity},
journal = {Fractals},
volume = {12},
number = {03},
pages = {331-346},
year = {2004},
doi = {10.1142/S0218348X04002586}
}

@ARTICLE{BianchiPianese2018,
	AUTHOR       = "Bianchi, S. and  Pianese, A.",
	TITLE        = "Time-varying {H}urst–{H}{\"o}lder exponents and the dynamics of (in)efficiency in stock markets",
	JOURNAL      = "Chaos, Solitons \& Fractals",
	YEAR         = "2018",
	VOLUME       = "109",
	PAGES        = "64-75"
}

@article{BrandiDiMatteo2022,
author = {Brandi, G. and Di Matteo, T.},
title = {Multiscaling and rough volatility: An empirical investigation},
journal = {International Review of Financial Analysis},
volume = {84},
pages = {102324},
year  = {2022}
}

@article{ComteRenault1998,
    title = {Long memory in continuous-time stochastic volatility models},
    author = {Comte, F. and Renault, E.},
    journal = {Mathematical Finance},
    volume = {8},
    number = {4},
    pages = {291-323},
    year = {1998}
}

@article{DiMatteo2007,
  title={Multi-scaling in finance},
  author={Di Matteo, T.},
  journal={Quantitative finance},
  volume={7},
  number={1},
  pages={21--36},
  year={2007},
  publisher={Taylor \& Francis}
}

@ARTICLE{Fama1970,
    AUTHOR       = "Fama, E.F.",
	TITLE        = "Efficient Capital Markets: a Review of Theory and Empirical Work",
	JOURNAL      = "The Journal of Finance",
	YEAR         = "1970",
	VOLUME       = "25",
	NUMBER       = "2",
	PAGES        = "383-417"
}

@article{GatheralJaissonRosenbaum2018,
  title={Volatility is rough},
  author={Gatheral, J. and Jaisson, T. and Rosenbaum, M.},
  journal={Quantitative Finance},
  volume={18},
  number={6},
  pages={933--949},
  year={2018},
  publisher={Taylor \& Francis}
}

@ARTICLE{Kolmogorov1940,
	AUTHOR       = "Kolmogorov, A.N.",
	TITLE        = "Wienershe Spiralen und einige andere interessante Kurven im Hilbertishen Raum",
	JOURNAL      = "DAS of the URSS (Nat. Sciences)",
	YEAR         = "1940",
	VOLUME       = "26",
	PAGES        = "115-118"}

@ARTICLE{MandelbrotVanNess1968,
	AUTHOR       = "Mandelbrot, B.B. and Van Ness, J.W.",
	TITLE        = "Fractional {B}rownian Motions, Fractional Noises and Applications",
	JOURNAL      = "SIAM Review",
	YEAR         = "1968",
	VOLUME       = "10",
	NUMBER 	     = "4",
	PAGES        = "422-437"}

@BOOK{MantegnaStanley2004,
	author = {Mantegna, R.N. and Stanley, H.E.},
	title = {An Introduction to Econophysics. Correlations and Complexity in Finance},
	publisher = {{Cambridge University Press}},
	year = {2004}}

@article{TaqquTeverovskyWillinger1995,
author = {Taqqu, M.S. and Teverovsky, V. and Willinger, W.},
title = {Estimators for Long-Range Dependence: {A}n empirical study},
journal = {Fractals},
volume = {3},
number = {4},
pages = {785-798},
year = {1995}}

@article{TeverovskyTaqqu1997,
  title={Testing for long-range dependence in the presence of shifting means or a slowly declining trend, using a variance-type estimator},
  author={Teverovsky, V. and Taqqu, M.},
  journal={Journal of Time Series Analysis},
  volume={18},
  number={3},
  pages={279--304},
  year={1997},
  publisher={Wiley Online Library}}

@article{WoodChan1994,
  title={Simulation of stationary {G}aussian processes in $[0, 1]^d$},
  author={Wood, A.T.A. and Chan, G.},
  journal={Journal of Computational and Graphical Statistics},
  volume={3},
  number={4},
  pages={409--432},
  year={1994},
  publisher={Taylor \& Francis}}

@article{Arcones1994,
  title={Limit theorems for nonlinear functionals of a stationary {G}aussian sequence of vectors},
  author={Arcones, M.A.},
  journal={Annals of Probability},
  volume={22},
  number={4},
  pages={2242--2274},
  year={1994}
}

@article{dehling1989,
  title={The empirical process of some long-range dependent sequences with an application to {U}-statistics},
  author={Dehling, H. and Taqqu, M.S.},
  journal={Annals of Statistics},
  volume={17},
  number={4},
  pages={1767--1783},
  year={1989}
}

@book{giraitis2012large,
  title={Large sample inference for long memory processes},
  author={Giraitis, L. and Koul, H.L. and Surgailis, D.},
  year={2012},
  publisher={World Scientific}
}

@book{OuannasEtal2023,
author = {Ouannas, A. and Batiha, I.M. and Pham, V.-T.},
title = {Fractional Discrete Chaos},
publisher = {World Scientific},
year = {2023},
doi = {10.1142/13277},
address = {},
edition   = {},
URL = {https://www.worldscientific.com/doi/abs/10.1142/13277},
eprint = {https://www.worldscientific.com/doi/pdf/10.1142/13277}
}

@article{carbone2004time,
  title={Time-dependent {H}urst exponent in financial time series},
  author={Carbone, A. and Castelli, G. and Stanley, H.E.},
  journal={Physica A: Statistical Mechanics and its Applications},
  volume={344},
  number={1-2},
  pages={267--271},
  year={2004},
  publisher={Elsevier}
}

@article{cheridito2003b,
  title={Fractional {O}rnstein-{U}hlenbeck processes.},
  author={Cheridito, P. and Kawaguchi, H. and Maejima, M.},
  journal={Electronic Journal of Probability [electronic only]},
  volume={8},
  year={2003}
}

@article{wang2023modeling,
  title={Modeling and forecasting realized volatility with the fractional {O}rnstein--{U}hlenbeck process},
  author={Wang, X. and Xiao, W. and Yu, J.},
  journal={Journal of Econometrics},
  volume={232},
  number={2},
  pages={389--415},
  year={2023},
  publisher={Elsevier}
}

@article{taqqu1975weak,
  title={Weak convergence to fractional {B}rownian motion and to the {R}osenblatt process},
  author={Taqqu, M.S.},
  journal={Advances in Applied Probability},
  volume={7},
  number={2},
  pages={249--249},
  year={1975},
  publisher={Cambridge University Press}
}

@article{dobrushin1979non,
  title={Non-central limit theorems for non-linear functional of {G}aussian fields},
  author={Dobrushin, R.L. and Major, P.},
  journal={Zeitschrift f{\"u}r Wahrscheinlichkeitstheorie und verwandte Gebiete},
  volume={50},
  number={1},
  pages={27--52},
  year={1979},
  publisher={Springer}
}

@article{breuer1983central,
  title={Central limit theorems for non-linear functionals of {G}aussian fields},
  author={Breuer, P. and Major, P.},
  journal={Journal of Multivariate Analysis},
  volume={13},
  number={3},
  pages={425--441},
  year={1983},
  publisher={Elsevier}
}

@book{podlubny1998fractional,
  title={Fractional {D}ifferential {E}quations},
  author={Podlubny, I.},
  volume={198},
  year={1998},
  publisher={Elsevier}
}

@book{brockwell2009time,
  title={Time series: {T}heory and methods},
  author={Brockwell, P.J. and Davis, R.A.},
  year={2009},
  publisher={Springer science \& business media}
}

@article{Weron2002,
  author  = {Weron, R.},
  title   = {Estimating long-range dependence: Finite sample properties and confidence intervals},
  journal = {Physica A: Statistical Mechanics and its Applications},
  volume  = {312},
  number  = {1--2},
  pages   = {285--299},
  year    = {2002},
  doi     = {10.1016/S0378-4371(02)00961-5}
}

@article{Kantelhardt2001,
  author  = {Kantelhardt, J.W. and Koscielny-Bunde, E. and Rego, H.H.A. and Havlin, S. and Bunde, A.},
  title   = {Detecting long-range correlations with detrended fluctuation analysis},
  journal = {Physica A: Statistical Mechanics and its Applications},
  volume  = {295},
  pages   = {441--454},
  year    = {2001},
  doi     = {10.1016/S0378-4371(01)00144-3}
}

\appendix

\section{Proof of the short-memory benchmark}\label{app_A}
We prove Proposition~\ref{prop:SRD} in the balanced regime $n,m\to\infty, \frac{n}{m}\to 1.$ Let 
\begin{equation}\nonumber
\xi_i(x) := \ind_{\{X_i\le x\}}-\Phi(x), \qquad \eta_j(x) := \ind_{\{Y_j\le x\}}-\Phi(x),
\end{equation}
so that 
\begin{equation}\nonumber
F_n(x)-G_m(x) = \frac1n\sum_{i=1}^n \xi_i(x) - \frac1m\sum_{j=1}^m \eta_j(x). 
\end{equation}
Define the marginal empirical processes
\begin{equation}\nonumber
\alpha_n(x) := \frac1{\sqrt n}\sum_{i=1}^n \xi_i(x), \qquad \beta_m(x) := \frac1{\sqrt m}\sum_{j=1}^m \eta_j(x). 
\end{equation}
Then 
\begin{equation}\nonumber
\sqrt{\frac{nm}{n+m}} \left( F_n(x)-G_m(x) \right) = \sqrt{\frac{m}{n+m}}\,\alpha_n(x) - \sqrt{\frac{n}{n+m}}\,\beta_m(x). 
\end{equation}
Since $n/m\to 1$, we have $\sqrt{\frac{m}{n+m}}\to \frac1{\sqrt2},$ $\sqrt{\frac{n}{n+m}}\to \frac1{\sqrt2}.$
Therefore, 
\begin{equation}\nonumber
\sqrt{\frac{nm}{n+m}} \left( F_n(x)-G_m(x) \right) = \frac1{\sqrt2} \left( \alpha_n(x)-\beta_m(x) \right) +o_P(1). 
\end{equation}
We now identify the covariance structure of the limiting process. If $U$ and $V$ are standard Gaussian random variables with correlation $\rho$, then 
\begin{equation}\nonumber
\operatorname{Cov} \left( \ind_{\{U\le x\}}, \ind_{\{V\le y\}} \right) = \Phi_2(x,y;\rho)-\Phi(x)\Phi(y), 
\end{equation}
where $\Phi_2(\cdot,\cdot;\rho)$ denotes the distribution function of a centered bivariate normal vector with unit variances and correlation $\rho$. Hence, writing $\rho_X(k)=\operatorname{Corr}(X_0,X_k)$, we define 
\begin{equation}\nonumber
\Gamma_X(x,y) := \sum_{k\in\mathbb Z} \left[ \Phi_2(x,y;\rho_X(k))-\Phi(x)\Phi(y) \right] \qquad \Gamma_Y(x,y) := \sum_{k\in\mathbb Z} \left[ \Phi_2(x,y;\rho_Y(k))-\Phi(x)\Phi(y) \right],
\end{equation}
and the cross-covariance kernels are 
\begin{equation}\nonumber
\Gamma_{XY}(x,y) := \sum_{k\in\mathbb Z} \left[ \Phi_2(x,y;\rho_{XY}(k))-\Phi(x)\Phi(y) \right], \qquad \Gamma_{YX}(x,y) := \sum_{k\in\mathbb Z} \left[ \Phi_2(x,y;\rho_{YX}(k))-\Phi(x)\Phi(y) \right].
\end{equation}
By assumption, the auto-covariance and cross-covariance sequences of the underlying Gaussian increments are absolutely summable. Since the functions $z\mapsto \ind_{\{z\le x\}}-\Phi(x)$ are bounded and measurable, the corresponding covariance series defining $\Gamma_X,\Gamma_Y,\Gamma_{XY},\Gamma_{YX}$ are finite. Therefore the joint empirical process $(\alpha_n,\beta_m)$ satisfies a functional central limit theorem in $\ell^\infty(\mathbb R)\times \ell^\infty(\mathbb R).$ Thus 
\begin{equation}\nonumber
(\alpha_n,\beta_m) \Rightarrow (\alpha,\beta), \end{equation} 
where $(\alpha,\beta)$ is a centered Gaussian pair with covariance kernels 
\begin{equation}\nonumber
    \begin{cases}
        \operatorname{Cov}(\alpha(x),\alpha(y)) = \Gamma_X(x,y),\\
        \operatorname{Cov}(\beta(x),\beta(y)) = \Gamma_Y(x,y),\\
        \operatorname{Cov}(\alpha(x),\beta(y)) = \Gamma_{XY}(x,y),\\
        \operatorname{Cov}(\beta(x),\alpha(y)) = \Gamma_{YX}(x,y).
    \end{cases}
\end{equation}
Consequently, 
\begin{equation}\nonumber
\sqrt{\frac{nm}{n+m}} \left( F_n-G_m \right) \Rightarrow U \qquad \text{in } \ell^\infty(\mathbb R), 
\end{equation}
where $U(x) = \frac1{\sqrt2} \left( \alpha(x)-\beta(x) \right)$. The covariance kernel of $U$ is therefore
\begin{equation}\nonumber
    \operatorname{Cov}(U(x),U(y)) = \frac12\Gamma_X(x,y) + \frac12\Gamma_Y(x,y) - \frac12 \left[ \Gamma_{XY}(x,y)+\Gamma_{YX}(x,y) \right].
\end{equation}
Finally, by the continuous mapping theorem applied to the supremum norm, 
\begin{equation}\nonumber
\sqrt{\frac{nm}{n+m}}D_{n,m} = \sup_{x\in\mathbb R} \left| \sqrt{\frac{nm}{n+m}} \left( F_n(x)-G_m(x) \right) \right| \Rightarrow \sup_{x\in\mathbb R}|U(x)|. \end{equation}
This proves Proposition~\ref{prop:SRD}.

\section{Proof of the long-memory benchmark}\label{app_B}
We prove Proposition~\ref{prop:LRD} in the balanced regime $n,m\to\infty, \;\frac{n}{m}\to 1.$ Throughout this Appendix we assume $H>1/2$. Let
\begin{equation}\nonumber
h_x(z) := \ind_{\{z\le x\}}-\Phi(x), \qquad x\in\mathbb R.
\end{equation}
For $Z\sim \mathcal N(0,1)$, we have 
\begin{equation}\nonumber
\mathbb E[h_x(Z)Z] = \mathbb E\left[ Z\ind_{\{Z\le x\}} \right] = \int_{-\infty}^x z\phi(z)\,dz = -\phi(x).
\end{equation}
Hence $h_x(Z)$ admits the Gaussian projection decomposition 
\begin{equation}\label{eq:indicator1}
h_x(Z) = -\phi(x)Z+r_x(Z), 
\end{equation}
where $r_x(Z) := h_x(Z)+\phi(x)Z$ satisfies $\mathbb E[r_x(Z)]=0$, and $\mathbb E[r_x(Z)Z]=0.$

Applying the identity \eqref{eq:indicator1} to the first sample gives 
\begin{equation}\nonumber
\sum_{i=1}^n h_x(X_i) = -\phi(x)\sum_{i=1}^n X_i + \sum_{i=1}^n r_x(X_i).
\end{equation}
Therefore, 
\begin{equation}\nonumber
n^{1-H} \left( F_n(x)-\Phi(x) \right) = -\phi(x)n^{-H}\sum_{i=1}^n X_i + n^{-H}\sum_{i=1}^n r_x(X_i).
\end{equation}
Similarly, 
\begin{equation}\nonumber
m^{1-H} \left( G_m(x)-\Phi(x) \right) = -\phi(x)m^{-H}\sum_{j=1}^m Y_j + m^{-H}\sum_{j=1}^m r_x(Y_j).
\end{equation} 
We first consider the leading linear terms. Since the increments are Gaussian and long-range dependent with parameter $H>1/2$, the normalized linear partial sums satisfy 
\begin{equation}\nonumber
n^{-H}\sum_{i=1}^n X_i \Rightarrow Z_X, \quad \text{ and }\quad m^{-H}\sum_{j=1}^m Y_j \Rightarrow Z_Y,
\end{equation}
where $(Z_X,Z_Y)$ is a centered Gaussian vector with 
\begin{equation}\nonumber
\operatorname{Var}(Z_X)=\sigma_X^2, \qquad \operatorname{Var}(Z_Y)=\sigma_Y^2, \qquad \operatorname{Cov}(Z_X,Z_Y)=\sigma_{XY}.
\end{equation}
It remains to show that the residual terms are negligible at the normalization used above. Let $U$ and $V$ be standard Gaussian random variables with correlation $\rho$. Then 
\begin{equation}\nonumber
\operatorname{Cov} \left( \ind_{\{U\le x\}}, \ind_{\{V\le y\}} \right) = \Phi_2(x,y;\rho)-\Phi(x)\Phi(y). 
\end{equation}
Using the decomposition in \eqref{eq:indicator} we obtain 
\begin{equation}\nonumber
\operatorname{Cov}(r_x(U),r_y(V)) = \Phi_2(x,y;\rho)-\Phi(x)\Phi(y) - \rho\,\phi(x)\phi(y).
\end{equation}
Moreover, 
\begin{equation}\nonumber
\left. \frac{\partial}{\partial \rho} \Phi_2(x,y;\rho) \right|_{\rho=0} = \phi(x)\phi(y).
\end{equation}
Thus the first-order term in $\rho$ cancels, and for $|\rho|$ small, $\operatorname{Cov}(r_x(U),r_y(V)) = \mathcal O(\rho^2).$ For fGn in the long-memory regime,
\begin{equation}\nonumber
\rho_X(k) \sim C_X |k|^{2H-2} \quad \text{ and }\quad \rho_Y(k) \sim C_Y |k|^{2H-2}, \qquad |k|\to\infty. 
\end{equation}
Hence 
\begin{equation}\nonumber
\operatorname{Cov}(r_x(X_0),r_y(X_k)) = \mathcal O\left(|k|^{4H-4}\right),
\end{equation}
and analogously for the $Y$-sample. Consequently, 
\begin{equation}\nonumber
\operatorname{Var} \left( n^{-H}\sum_{i=1}^n r_x(X_i) \right) \le C n^{-2H} \sum_{i=1}^n\sum_{j=1}^n |\rho_X(i-j)|^2. 
\end{equation}
The right-hand side converges to zero. More precisely, 
\begin{equation}\nonumber
\left\| n^{-H}\sum_{i=1}^n r_x(X_i) \right\|_{L^2} = \begin{cases} \mathcal O\left(n^{1/2-H}\right), & 1/2<H<3/4, \\[4pt] \mathcal O\left(n^{-1/4}\sqrt{\log n}\right), & H=3/4, \\[4pt] \mathcal O\left(n^{H-1}\right), & 3/4<H<1. \end{cases} 
\end{equation}
The same estimates hold for the $Y$-sample, with $m$ in place of $n$. Using the monotonicity of the empirical distribution functions, together with standard discretization and Gaussian tail bounds, the previous pointwise bounds extend to the supremum over $x\in\mathbb R$. Therefore,
\begin{equation}\nonumber \sup_{x\in\mathbb R} \left| n^{-H}\sum_{i=1}^n r_x(X_i) \right| \to 0 \quad \text{ and } \quad \sup_{x\in\mathbb R} \left| m^{-H}\sum_{j=1}^m r_x(Y_j) \right| \to 0 \qquad \text{in probability}.
\end{equation}
It follows that, in $\ell^\infty(\mathbb R)$, $n^{1-H} \left( F_n-\Phi \right) \Rightarrow -\phi(\cdot)Z_X,$ and $m^{1-H} \left( G_m-\Phi \right) \Rightarrow -\phi(\cdot)Z_Y$. 

Now set $c_{n,m} := \frac{n^{1-H}m^{1-H}} {n^{1-H}+m^{1-H}}$. Then
\begin{equation}\nonumber 
c_{n,m} \left( F_n(x)-G_m(x) \right) = c_{n,m}n^{H-1} \left[ n^{1-H} \left( F_n(x)-\Phi(x) \right) \right] - c_{n,m}m^{H-1} \left[ m^{1-H} \left( G_m(x)-\Phi(x) \right) \right]. 
\end{equation}
Since $n/m\to1$, we have 
\begin{equation}\nonumber
c_{n,m}n^{H-1}\to \frac12, \qquad c_{n,m}m^{H-1}\to \frac12. 
\end{equation}
Therefore, 
\begin{equation}\nonumber
c_{n,m} \left( F_n(x)-G_m(x) \right) \Rightarrow -\phi(x)Z_0,
\end{equation} 
where $Z_0 = \frac12 \left( Z_X-Z_Y \right).$ The random variable $Z_0$ is centered Gaussian with variance 
\begin{equation}\nonumber \sigma_0^2 = \frac14 \left( \sigma_X^2+\sigma_Y^2-2\sigma_{XY} \right). 
\end{equation}
Finally, by the continuous mapping theorem, 
\begin{equation}\nonumber
D_{n,m}^{\diamond}=c_{n,m}D_{n,m} = \sup_{x\in\mathbb R} \left| c_{n,m} \left( F_n(x)-G_m(x) \right) \right| \Rightarrow \sup_{x\in\mathbb R} \phi(x)|Z_0|.
\end{equation} 
Since $\sup\limits_{x\in\mathbb R}\phi(x) = \phi(0) = \frac1{\sqrt{2\pi}},$ we obtain 
\begin{equation}\nonumber
D_{n,m}^{\diamond}:=c_{n,m}D_{n,m} \Rightarrow \frac1{\sqrt{2\pi}}|Z_0|.
\end{equation}
This proves Proposition~\ref{prop:LRD}.

\end{document}